\documentclass[12pt]{amsart}

\usepackage[utf8]{inputenc}
\usepackage{amsmath,amssymb}
\usepackage{thm-restate}
\usepackage{enumerate}
\usepackage{framed}
\usepackage{geometry}
\usepackage[normalem]{ulem}
\usepackage[foot]{amsaddr}

\usepackage[unicode=true,pdfusetitle,
 bookmarks=true,bookmarksnumbered=false,bookmarksopen=false,
 breaklinks=false,pdfborder={0 0 1},backref=false,colorlinks=false,bookmarksdepth=3]
 {hyperref}
\usepackage[usenames,dvipsnames]{xcolor}
\hypersetup{colorlinks=true, linkcolor=Maroon, citecolor=OliveGreen, filecolor=magenta, urlcolor=Blue}

\newcommand{\eq}[1]{\hyperref[eq:#1]{(\ref*{eq:#1})}}
\renewcommand{\sec}[1]{\hyperref[sec:#1]{Section~\ref*{sec:#1}}}
\newcommand{\app}[1]{\hyperref[app:#1]{Appendix~\ref*{app:#1}}}
\newcommand{\theo}[1]{\hyperref[thm:#1]{Theorem~\ref*{thm:#1}}}
\newcommand{\algo}[1]{\hyperref[alg:#1]{Algorithm~\ref*{alg:#1}}}
\newcommand{\lemm}[1]{\hyperref[lem:#1]{Lemma~\ref*{lem:#1}}}
\newcommand{\defin}[1]{\hyperref[defn:#1]{Definition~\ref*{defn:#1}}}
\newcommand{\corr}[1]{\hyperref[cor:#1]{Corollary~\ref*{cor:#1}}}
\newcommand{\fig}[1]{\hyperref[fig:#1]{Figure~\ref*{fig:#1}}}
\newcommand{\propos}[1]{\hyperref[prop:#1]{Proposition~\ref*{prop:#1}}}
\newcommand{\rema}[1]{\hyperref[rem:#1]{Remark~\ref*{rem:#1}}}

\newcommand{\prim}{primitive}

\newtheorem{thm}{Theorem}[section]
\newtheorem{lem}[thm]{Lemma}
\newtheorem{prop}[thm]{Proposition}
\newtheorem{cor}[thm]{Corollary}
\theoremstyle{definition}
\newtheorem{dfn}[thm]{Definition}
\theoremstyle{remark}
\newtheorem{rem}[thm]{Remark}

\usepackage{algpseudocode}

\algnewcommand{\A}{\textbf{and}\space}
\algnewcommand{\Or}{\textbf{or}\space}
\algnewcommand{\Xor}{\textbf{xor}\space}
\makeatletter
\let\OldStatex\Statex
\renewcommand{\Statex}[1][3]{%
  \setlength\@tempdima{\algorithmicindent}%
  \OldStatex\hskip\dimexpr#1\@tempdima\relax}
\makeatother

\author{Vadym Kliuchnikov$^{1}$\email{vadym@microsoft.com}}
\author{Jon Yard$^{1}$\email{jonyard@microsoft.com}}
\address{$^1$ Quantum Architectures and Computation Group, Microsoft Research, Redmond, WA, USA}

\def\z{\mathbb{Z}}

\def\q{\mathbb{Q}}

\def\r{\mathbb{R}}

\def\floor#1{\mathopen{}\left\lfloor #1\right\rfloor\mathclose{}}

\def\at#1{\mathopen{}\left(#1\right)\mathclose{}}
\def\of#1{\mathopen{}\left[#1\right]\mathclose{}}
\def\set#1{\mathopen{}\left\{  #1\right\}\mathclose{}}
\def\tpl#1{\mathopen{}\left\langle #1\right\rangle\mathclose{}}

\def\ket#1{\mathopen{}\left|#1\right\rangle\mathclose{}}
\def\bra#1{\mathopen{}\left\langle #1\right|\mathclose{}}

\def\UG#1#2#3#4{{\at{\begin{array}{cc} #1 & #2 \\ #3 & #4 \end{array} }}}

\DeclareMathOperator{\nrd}{Nrd}

\DeclareMathOperator{\ord}{ord}

\def\leg#1#2{{#1 \overwithdelims () #2}}

\newcommand{\nc}{\newcommand}

\DeclareMathOperator{\Tr}{Tr}
\DeclareMathOperator{\Trd}{Trd}
\DeclareMathOperator{\Nm}{N}

\DeclareMathOperator{\disc}{disc}
\DeclareMathOperator{\rad}{rad}

\makeatletter
\let\O\@undefined
\def\O{\mathrm{O}}

\def\SU{\mathrm{SU}}

\def\leg#1#2{{#1 \overwithdelims () #2}}
\def\ox{\otimes}

\nc{\CA}{\mathcal{A}} \nc{\CB}{\mathcal{B}} \nc{\CC}{\mathcal{C}}
\nc{\CD}{\mathcal{D}} \nc{\CE}{\mathcal{E}} \nc{\CF}{\mathcal{F}}
\nc{\CG}{\mathcal{G}} \nc{\CH}{\mathcal{H}} \nc{\CI}{\mathcal{I}}
\nc{\CJ}{\mathcal{J}} \nc{\CK}{\mathcal{K}} \nc{\CL}{\mathcal{L}}
\nc{\CM}{\mathcal{M}} \nc{\CN}{\mathcal{N}} \nc{\CO}{\mathcal{O}}
\nc{\CP}{\mathcal{P}} \nc{\CQ}{\mathcal{Q}} \nc{\CR}{\mathcal{R}} 
\nc{\CS}{\mathcal{S}} \nc{\CT}{\mathcal{T}} \nc{\CU}{\mathcal{U}} 
\nc{\CV}{\mathcal{V}} \nc{\CW}{\mathcal{W}} \nc{\CX}{\mathcal{X}} 
\nc{\CY}{\mathcal{Y}} \nc{\CZ}{\mathcal{Z}}

\nc{\bA}{\mathbb{A}} \nc{\bB}{\mathbb{B}} \nc{\bC}{\mathbb{C}}
\nc{\bD}{\mathbb{D}} \nc{\bE}{\mathbb{E}} \nc{\bF}{\mathbb{F}}
\nc{\bG}{\mathbb{G}} \nc{\bH}{\mathbb{H}} \nc{\bI}{\mathbb{I}}
\nc{\bJ}{\mathbb{J}} \nc{\bK}{\mathbb{K}} \nc{\bL}{\mathbb{L}}
\nc{\bM}{\mathbb{M}} \nc{\bN}{\mathbb{N}} \nc{\bO}{\mathbb{O}}
\nc{\bP}{\mathbb{P}} \nc{\bQ}{\mathbb{Q}} \nc{\bR}{\mathbb{R}} 
\nc{\bS}{\mathbb{S}} \nc{\bT}{\mathbb{T}} \nc{\bU}{\mathbb{U}} 
\nc{\bV}{\mathbb{V}} \nc{\bW}{\mathbb{W}} \nc{\bX}{\mathbb{X}}
\nc{\bZ}{\mathbb{Z}}

\nc{\BA}{\mathbf{A}} \nc{\BB}{\mathbf{B}} \nc{\BC}{\mathbf{C}}
\nc{\BD}{\mathbf{D}} \nc{\BE}{\mathbf{E}} \nc{\BF}{\mathbf{F}}
\nc{\BG}{\mathbf{G}} \nc{\BH}{\mathbf{H}} \nc{\BI}{\mathbf{I}}
\nc{\BJ}{\mathbf{J}} \nc{\BK}{\mathbf{K}} \nc{\BL}{\mathbf{L}}
\nc{\BM}{\mathbf{M}} \nc{\BN}{\mathbf{N}} \nc{\BO}{\mathbf{O}}
\nc{\BP}{\mathbf{P}} \nc{\BQ}{\mathbf{Q}} \nc{\BR}{\mathbf{R}} 
\nc{\BS}{\mathbf{S}} \nc{\BT}{\mathbf{T}} \nc{\BU}{\mathbf{U}} 
\nc{\BV}{\mathbf{V}} \nc{\BW}{\mathbf{W}} \nc{\BX}{\mathbf{X}} 
\nc{\BY}{\mathbf{Y}} \nc{\BZ}{\mathbf{Z}}

\nc{\msA}{\mathscr{A}} \nc{\msB}{\mathscr{B}} \nc{\msC}{\mathscr{C}}
\nc{\msD}{\mathscr{D}} \nc{\msE}{\mathscr{E}} \nc{\msF}{\mathscr{F}}
\nc{\msG}{\mathscr{G}} \nc{\msH}{\mathscr{H}} \nc{\msI}{\mathscr{I}}
\nc{\msJ}{\mathscr{J}} \nc{\msK}{\mathscr{K}} \nc{\msL}{\mathscr{L}}
\nc{\msM}{\mathscr{M}} \nc{\msN}{\mathscr{N}} \nc{\msO}{\mathscr{O}}
\nc{\msP}{\mathscr{P}} \nc{\msQ}{\mathscr{Q}} \nc{\msR}{\mathscr{R}} 
\nc{\msS}{\mathscr{S}} \nc{\msT}{\mathscr{T}} \nc{\msU}{\mathscr{U}} 
\nc{\msV}{\mathscr{V}} \nc{\msX}{\mathscr{X}} \nc{\msW}{\mathscr{W}} 
\nc{\msY}{\mathscr{Y}} \nc{\msZ}{\mathscr{Z}}

\nc{\mfa}{{\mathfrak a}} \nc{\mfb}{{\mathfrak b}} \nc{\mfc}{{\mathfrak c}}
\nc{\mfd}{{\mathfrak d}} \nc{\mfe}{{\mathfrak e}} \nc{\mff}{{\mathfrak f}}
\nc{\mfg}{{\mathfrak g}} \nc{\mfh}{{\mathfrak h}} \nc{\mfi}{{\mathfrak i}}
\nc{\mfj}{{\mathfrak j}} \nc{\mfk}{{\mathfrak k}} \nc{\mfl}{{\mathfrak l}}
\nc{\mfm}{{\mathfrak m}} \nc{\mfn}{{\mathfrak n}} \nc{\mfo}{{\mathfrak o}}
\nc{\mfp}{{\mathfrak p}} \nc{\mfq}{{\mathfrak q}} \nc{\mfr}{{\mathfrak r}}
\nc{\mfs}{{\mathfrak s}} \nc{\mft}{{\mathfrak t}} \nc{\mfu}{{\mathfrak u}}
\nc{\mfv}{{\mathfrak v}} \nc{\mfw}{{\mathfrak w}} \nc{\mfx}{{\mathfrak x}}
\nc{\mfy}{{\mathfrak y}} \nc{\mfz}{{\mathfrak z}}

\nc{\mfA}{{\mathfrak A}} \nc{\mfB}{{\mathfrak B}} \nc{\mfC}{{\mathfrak C}}
\nc{\mfD}{{\mathfrak D}} \nc{\mfE}{{\mathfrak E}} \nc{\mfF}{{\mathfrak F}}
\nc{\mfG}{{\mathfrak G}} \nc{\mfH}{{\mathfrak H}} \nc{\mfI}{{\mathfrak I}}
\nc{\mfJ}{{\mathfrak J}} \nc{\mfK}{{\mathfrak K}} \nc{\mfL}{{\mathfrak L}}
\nc{\mfM}{{\mathfrak M}} \nc{\mfN}{{\mathfrak N}} \nc{\mfO}{{\mathfrak O}}
\nc{\mfP}{{\mathfrak P}} \nc{\mfQ}{{\mathfrak Q}} \nc{\mfR}{{\mathfrak R}}
\nc{\mfS}{{\mathfrak S}} \nc{\mfT}{{\mathfrak T}} \nc{\mfU}{{\mathfrak U}}
\nc{\mfV}{{\mathfrak V}} \nc{\mfW}{{\mathfrak W}} \nc{\mfX}{{\mathfrak X}}
\nc{\mfY}{{\mathfrak Y}} \nc{\mfZ}{{\mathfrak Z}}

\def\emph#1{{\bf #1}}
\nc{\proj}[1]{\ket{#1}\bra{#1}}
\nc{\braket}[2]{{\langle #1 | #2 \rangle}}

\def\Ann{\mathrm{Ann}}

\def\al{\alpha}
\def\ph{\varphi}
\def\dist{\mathrm{dist}} 

\def\Adj{\mathrm{Adj}} 
\def\nf{F} 
\def\zf{R_F} 
\def\qa{Q}
\def\mo{\mathcal{M}} 
\def\mos{\mathcal{M}_S} 
\def\ids{S} 
\def\tsm{T_{S}\at{\mo}} 
\def\gs{G_{S}\at{\mo}} 
\def\p{\mathfrak{p}} 
\def\P{P} 
\def\a{\mathfrak{a}} 
\def\Adj{\mathrm{Adj}}
\def\gens{\mathrm{gen}_S\at{\mo}} 
\def\ugens{\mathrm{gen}_{u}\at{\mo}} 
\def\gext{\mathrm{ext}_S\at{\mo}} 
\def\gone{\mathrm{ext}^{1}_S\at{\mo}} 
\def\gtwo{\mathrm{ext}^{2}_S\at{\mo}} 
\def\gtwoext{\mathrm{ext}^{2}_{S_{ext}}\at{\mo}} 

\def\ordr#1{O_r\at{#1}} 
\def\ordl#1{O_\ell\at{#1}} 

\def\cmi#1{\mu_{S}\at{#1}} 
\def\pone{\mathbb{P}^{1}} 

\def\i{\mathbf{i}}
\def\j{\mathbf{j}}
\def\k{\mathbf{k}}

\def\ve{\varepsilon} 

\begin{document}
\bibliographystyle{plainurl}

\title{A framework for exact synthesis}

\begin{abstract} Exact synthesis is a tool used in algorithms for approximating an arbitrary qubit unitary with a sequence of quantum gates from some finite set. These approximation algorithms find asymptotically optimal approximations in probabilistic polynomial time, in some cases even finding the optimal solution in probabilistic polynomial time given access to an oracle for factoring integers. In this paper, we present a common mathematical structure underlying all results related to the exact synthesis of qubit unitaries known to date, including Clifford+T, Clifford-cyclotomic and V-basis gate sets, as well as gates sets induced by the braiding of Fibonacci anyons in topological quantum computing.   The framework presented here also provides a means to answer questions related to the exact synthesis of unitaries for wide classes of other gate sets, such as Clifford+T+V and $\SU(2)_k$ anyons. 
\end{abstract}

\maketitle


\section{Introduction}

The exact synthesis of unitaries is an important tool for compiling quantum algorithms into a sequence of individual quantum gates. Compiling a quantum algorithm for a given quantum computer architecture includes expressing or rewriting the algorithm using the elementary operations supported by target quantum computer architecture. A typical quantum computer architecture has a fault tolerance layer implemented in it. Fault tolerance can be achieved on a logical level, by using fault-tolerant protocols based on quantum error correcting codes, on a physical level, by computing with topologically-protected degrees of freedom (topological quantum computers), or by using some combination of both. The set of elementary operations supported on a quantum computer typically includes a finite set of single-qubit unitaries, a finite set of two-qubit unitaries, a finite set of single-qubit states that can be prepared, and a finite set of measurements that can be performed.  We refer to the set of single-qubit unitaries as the single-qubit gate set supported by a given architecture. Quantum algorithms are usually designed and expressed in terms of multi-qubit unitaries that must be expressed using the available one- and two-qubit gates.
Below, we give a overview of methods for compiling single-qubit unitaries and highlight the role of exact synthesis in this question. For more detailed overview, the reader may consult e.g.\ Chapter 2 in \cite{TH}. 

One of the basic tasks in quantum compiling is to implement an arbitrary single qubit unitary on the target quantum computer architecture. One of the most common types of single qubit unitaries appearing in applications is $R_z\at{\phi}=e^{i\phi Z/2}$ where $Z$ is a Pauli Z matrix. However, a computer with finitely-many single-qubit operations can never exactly implement every conceivable angle $\phi$.  

For most $\phi$ that appear in applications, the corresponding unitary $R_z\at{\phi}$ can not be expressed exactly using the finitely-many available operations. Therefore, one must be content to approximate $R_z\at{\phi}$ to within some precision $\ve$. There are two common approaches to this problem. One approach involves using ancillary qubits, special states and measurements \cite{WK,DCS,FSt,J,bk:ksv,KMM2,RUS2,RUS1}. The other approach uses only the single-qubit operations available to a given target architecture  for approximating $R_z\at{\phi}$ \cite{KMM3,KBS,S,BGS,RS,BBG,DN,RossV}. Here, we mainly focus on the second class. Frequently, it is possible to extend the original single-qubit gate set available on a given architecture using multi-qubit circuits with measurements and classical feedback that effectively implements a single-qubit unitary \cite{WK,Paetznick2013,RUS1,RUS2,BKZ}. For example, the V-basis unitaries can be implemented using such circuits on any architecture that support single-qubit Clifford+T unitaries and measurements in the computational basis. The Exact synthesis of single-qubit unitaries is an important part of recently-developed algorithms for finding approximations to single-qubit unitaries using single-qubit unitaries from the Clifford+T \cite{S,RS}, V-basis \cite{BGS,BBG,RossV} and Fibonacci anyon gate sets \cite{KBS}. A distinct feature of these algorithms is that they find asymptotically-optimal \footnote{By asymptotically optimal, we mean that the length of the circuit found by the algorithm is less then $C\cdot l_{opt} + C'$, where $l_{opt}$ is the length of the shortest circuit that approximates the target unitary with precision $\ve$. The constants $C,C'$ do not depend on the target unitary or on $\ve$. } approximations in probabilistic polynomial time \cite{S,KBS,BGS}, or optimal approximations in polynomial time, given access to an oracle for factoring integers \cite{RS,BBG,RossV}. All these algorithms exploit the number-theoretic structure of the problem. For brevity we will refer to them as Number Theoretic Unitary Approximation Algorithms (NTUAAs). Next we briefly outline the high-level idea behind NTUAAs and explain the precise role played by exact synthesis in their implementations.

Each NTUAA is designed for a specific target gate set (e.g.\ single-qubit Clifford+T, V-basis, Fibonacci) and for a specific distance measure $\rho$ defined on single quit unitaries (frequently the operator norm). Given a target unitary $U'$, the algorithm outputs a circuit $U = U_n U_{n-1} \cdots U_1$ of length $n = O\at{\log\at{1/\ve}}$ achieving $\dist(U,U') \leq \ve$, with each $U_i$ in the target gate set.  When we say that such an algorithm runs in probabilistic polynomial time or in polynomial time, we mean that the running time is bounded by a polynomial in $\log\at{1/\ve}$. The first stage of any NTUAAs to find a unitary $U$ that (a) can be represented exactly using the target gate set (b) is within distance $\ve$ from $U'$ and (c) can be implemented using at most $O\at{\log\at{1/\ve}}$ gates from the target gate set. The second stage is finding a circuit over the target gate set that implements $U$. This step is performed by a polynomial-time (in the number of bits needed to represent $U'$)  exact synthesis algorithm. There are two key aspects to exact synthesis: (1) a simple description of all unitaries that can be implemented over the target gate set using the number of gates bounded by some constant (2) a polynomial-time algorithm for compiling a unitary that is known to be representable over the gate set. Previous works gave answers to this question for single qubit Clifford+T gates, V-basis, Fibonacci gate set and Clifford+$R_z(\pi/n)$ for $n=6,8,12,16$. In this paper we develop a framework for answering questions (1) and (2) and several related ones for a wide range of gate sets using the theory of maximal orders of central simple algebras and quaternion algebras. Our goal here is to set up the foundation for generalizing the ideas behind NTUAAs to a wide range of gate sets. The application of our framework to approximating unitaries will be treated elsewhere~\cite{InPrep}. 

In the next section we briefly review the necessary definitions and notation. In \sec{factor} we state the questions relevant to exact synthesis for qubits in the language of quaternion algebras. In \sec{app} we show how to use tools from \sec{factor} to answer questions related to exact synthesis and discuss how previously known results can be derived using the formalism developed in this paper. 

\section{Preliminaries and notation}



Fix an integral domain $R$ with fraction field $K$ and let $A$ be a finite-dimensional $K$-algebra $A$.  An \emph{ideal} $I$ of $A$ is a torsion-free $R$-submodule of $A$ such that $I\ox_R K = A$, i.e.\ an $R$-sublattice of $A$. 
An ideal of $A$ that is also a ring (hence an $R$-algebra) is called an \emph{order} of $A$. Equivalently, an order is an ideal in which each element is $R$-integral, i.e.\ is a zero of a monic polynomial with coefficients from $R$. The set of orders is partially ordered by inclusion such that every order is contained in some maximal order.   If $A$ is commutative then contains a unique maximal order, equal to the integral closure of $R$ in $A$.  The maximal $\bZ$-order in a number field $F$ is its \emph{ring of integers}, denoted $R_F$ throughout.  If $A$ is noncommutative, the integral closure of $R$ in $A$ is not generally a ring, but is rather equal to the union of all maximal orders in $A$.  Maximal orders may be viewed as integral closures that are constrained to be rings.   
Given an ideal $I$ of $A$, the sets $\ordl{I} = \{a \in A : aI \subset I\}$ and $\ordr{I} = \{a \in A : I a \subset I\}$ are orders, respectively called the \emph{left order} and the \emph{right order} of $I$.  When $\ordl{I} = \CO_\ell$ and $\ordr{I} = \CO_r$, we may also call $I$ a  \emph{left $\CO_\ell$-ideal}, a \emph{right $\CO_r$-ideal} or a \emph{$\CO_\ell$-$\CO_r$-ideal}, and when $\CO_\ell = \CO_r = \CO$, we call $I$ a two-sided $\CO$-ideal.  Then $\CO_\ell$ is maximal iff $\CO_r$ is maximal, in which case we follow \cite{R} and call $I$ a \emph{normal ideal} of $A$.  Normal ideals are invertible on the left and on the right.  A normal ideal $I$ is contained in its left order iff it is contained in its right order, in which case we follow \cite{R} and call $I$ an \emph{integral ideal} of $A$.  
Each integral ideal $I$ of $A$ is simultaneously a left ideal of $\CO_\ell$ and a right ideal of $\CO_r$.
An integral ideal $I$ is a maximal left ideal of $\CO_\ell$ iff it is a maximal right ideal of $\CO_r$, in which case we call $I$ a \emph{maximal ideal}  (note that \cite{R} calls such ideals ``maximal integral ideals''). 
We call a normal ideal $I$ \emph{primitive} if it is not contained in any proper two-sided ideal (equivalently, prime ideal) of $\CO_\ell$ (equivalently, of $\CO_r$).   
The product $I_1 \cdots I_n$ of ideals of $A$ is \emph{proper} if $\ordl{I_j} = \ordr{I_{j+1}}$ for $j = 1,\dotsc,n-1$.


If $\CR$ is a ring, we write $\CR^\times$ for the group of units in $\CR$.  A \emph{prime ideal} of $\CR$ is a proper two-sided ideal $P$ of $\CR$ such that $IJ \subset P$ implies that either $I \subset P$ or $J\subset P$, where $I \subset \CR$ and $J \subset \CR$ are ideals (equivalently right ideals, or equivalently left ideals).  Equivalently, a two-sided ideal $P \subset \CR$ is prime if $(a)(b) \subset P$ implies that either $a \in P$ or $b\in P$, where $(a):= \CR a \CR$.  In a general ring $\CR$, every maximal ideal is prime but we will see that  for orders over Dedekind domains, every prime ideal is also maximal.

If $R$ is a Dedekind domain and $\CO$ is an $R$-order in a separable $K$-algebra, then a two-sided ideal of $\CO$ is prime iff it is maximal.  This is because maximal two-sided ideals are always prime, whereas the reverse direction follows from Theorem 22.3 of \cite{R}).  In this case, each prime ideal $\P$ of $\CO$ determines a unique prime ideal $\p$ of $R$ via 
\begin{equation}
\P \mapsto \p = \P\cap R.
\end{equation}

If $R$ is a Dedekind domain and $\CM$ is a maximal $R$-order in a central simple $K$-algebra, there is a 1-1 correspondence between prime ideals of $R$ and of $\CM$, with the map going the other way given by (Theorem 22.4 of \cite{R})
\begin{equation}
\p \mapsto \P = \CM \cap \rad(\CM_\mfp),    
\end{equation}
where $\CM_\mfp = (R - \mfp)^{-1} \CM$ is the localization of $\CM_\mfp$ at $\mfp$ and where  $\rad$ denotes the Jacobson radical (intersection of all maximal left ideals); in this case it is the unique maximal ideal
 $\rad(\CM_\mfp) = P \CM_\mfp = \CM_\mfp P$.  Theorem 22.10 of \cite{R} shows that the two-sided $\CM$-ideals form an abelian group whose elements factor uniquely into products of prime ideals of $\CM$ and their inverses.  

The structure of left ideals (or equivalently, right ideals) of a maximal order $\CM$ over a Dedekind domain $R$ in a central simple $K$-algebra is as follows.  
By Theorem 22.15 of \cite{R}, each maximal left ideal $M$ of $\CM$ contains exactly one prime ideal $P = \Ann_\CO(\CO/M) = \{x \in \CO : x\CO \subset M\}$ of $\CM$, corresponding to the prime ideal $\mfp = \Ann_R(\CO/M) = \{x \in R : x\CO \subset M\}$.

By a \emph{quaternion algebra} over a field $K$, we mean a central simple $K$-algebra of dimension 4.  Any quaternion algebra has an explicit description, in terms of some $a,b \in K^\times$, as the $K$-algebra 
\[\leg{a,b}{K}  = \{x_0 + x_1 \i + x_2 \j + x_3 \k : x_i \in F, \i^2 = a, \j^2 = b, \i \j = \k\}.\]
We denote the \emph{Hamilton quaternions} as $\bH = \leg{-1,-1}{\bR}$.
 The map $q \mapsto \bar q$ taking a quaternion $q= q_0 + q_1 \i + q_2 \j + q_3 \k$ to its \emph{conjugate} $\bar q = q_0 - q_1 \i - q_2 \j - q_3 \k$ is $K$-linear involution (i.e.\ an order-2 automorphism such that $\bar{q r} = \bar{r} \bar{q}$) and can be used to express its \emph{reduced norm} as $\nrd(q) = q\bar q = q_0^2 -a q_1^2 -b  q_2^2 + ab q_3^2$  and \emph{reduced trace} $\Trd(q) = q + \bar q = 2q_0$.
By the \emph{norm} of an ideal $I$ of $Q$, we mean  the fractional ideal of $K$ generated by the reduced norms of the elements of $I$.  By the \emph{discriminant} of an order $\CO$ in a quaternion algebra, we mean the square root of the ideal generated by the determinants of all possible $4\times 4$ matrices $(\Tr x_i x_j)_{ij}$ where $x_1,\dotsc,x_4 \in \CO$.  The \emph{discriminant} $\disc(Q)$ of a quaternion algebra $Q$ is then defined to be the discriminant of any maximal $R_F$-order in $Q$ (these are sometimes called ``reduced discriminants'' in the literature).  It is well-known that the discriminant of any quaternion algebra factors into a square-free product of finitely many prime ideals of $\zf$.  
 A quaternion algebra \emph{splits} at a place $v$ of $K$ if $Q_v = Q \ox_K K_v \simeq K_v^{2\times 2}$ and otherwise it is \emph{ramified} at $v$.  If $v$ is a finite place, then $Q$ splits at $v$ iff the corresponding prime ideal $\mfp_v$ of $K$ divides the discriminant of $Q$.   If $v$ is an infinite place of $K$, then $Q$ is ramified at $v$ if $Q_v \simeq \bH$ and it is split at $v$ if $Q_v \simeq \bR^{2\times 2}$.  If every real place of $K$ is ramified in $Q$, then $Q$ is called \emph{totally definite} and if it is not totally definite, we call it \emph{indefinite}.
Recall that each number field $K$ satisfies 
\[K \ox_\bQ \bR \simeq \prod_v K_v = \bR^r \times \bC^c,\]
where the product is over the infinite places $v$ of $K$, $r$ is the number of real places, each corresponding to an embeddings $K \hookrightarrow \bR$, $c$ is the number of complex places, each corresponding to a pair of distinct embeddings $K \hookrightarrow \bC$ related by complex conjugation, and where $[K:\bQ] = r +2c$.   Similarly, 
\[Q\ox_\bQ \bR \simeq \prod_v Q_v = \bH^{r_0} \times (\bR^{2\times 2})^{r_1} \times  (\bC^{2\times 2})^{c},\]
 where $r_0 + r_1= r$, so $Q$ is totally definite iff it cannot be embedded into $\bR^{2\times 2}$.

\section{Factorization of quaternions} \label{sec:factor}

Let $\nf$ be a number field with ring of integers $\zf$, let $\qa$ be a quaternion algebra over $\nf$, let $\mo\subset \qa$ be a maximal $\zf$-order and let $\ids$ be a finite set of prime ideals of $\zf$. Let $\mos$ be all elements of $\mo$ such that their $\nrd$ factors into elements of $S$. In this section we are discussing the following questions: 
\begin{enumerate}
	\item Is there a finite set $\CG$ of quaternions such that any element of $\mos$ can be written as $q_1 \ldots q_n$ for $q_k \in \CG$? 
	\item How to find $\CG$ if (1) is true ?
	\item Can the set $\CG$ in (1) be chosen to be a subset of $\mos$? 
	\item How to find $\CG \subset \mos $ if (3) is true? 
\end{enumerate} 

We provide answers to questions (1) and (2) for a wide class of quaternion algebras in \sec{factor-1}. In \sec{graph}, we describe an algorithm for deciding if (3) is true and for constructing $\CG \subset \mos $. We also discuss how to estimate the size of  $\CG \subset \mos $ algorithmically without computing it. The mathematical tools for answering the above questions were developed in \cite{R,KV,V} and many other works. Our main contribution is the explanation of how to apply them to the stated questions. 

\subsection{   Factorization of elements of $\mos$ into a finite set of quaternions }  \label{sec:factor-1}

We keep the notation described in the beginning of \sec{factor}. The main result of this section is the following: 

\begin{thm} \label{thm:factor-1} Any element of $\mos$ can be written as a product $q_1 \ldots q_n u \alpha$ where quaternions $q_1,\ldots,q_n$ are from a finite set $\gext$, $u \in \CM^\times$ and $\alpha \in F^\times$. If one of the following conditions holds, then $u$ can be written as the product of a finite number elements of $\ugens$: 
\begin{enumerate}
	\item $F$ is a totally real number field, and $\qa$ is a totally definite quaternion algebra. 
	\item $F$ is a totally real number field, and $\qa$ is split in exactly one real place.  
\end{enumerate}
There are algorithms for finding $\gext$, $\ugens$ and factorization $q_1 \ldots q_n u \alpha$. 
\end{thm}

In other words, when the conditions $(1)$ or $(2)$ mentioned in the theorem above hold, any element of $\mos$ can be written as a finite product of elements of $\gext \cup \ugens$. We postpone the proof of the above theorem until the end of this section. We will first describe set $\gext$ and provide an algorithm for computing it.  Next we discuss how to compute the set $\ugens$ and finally provide the algorithm for finding the factorization described in the theorem. The theorem will immediately follow from the proof of the correctness of the algorithm.

Instead of studying the factorization of an element $q$ of $\mos$ we will study the factorization of the right ideal $q\mo$ into maximal right ideals of $\CM$. The following result is crucial for this: 

\begin{thm}[Special case of Theorem 22.18 in \cite{R}]\label{thm:factor} Let $\mo$ be a maximal $\zf$-order in a quaternion algebra $\qa$ and let $I$ be a left ideal of $\CM$ such that the left $\CM$-module $\mo/I$ has composition length $n$. Then $I$ is expressible as a proper product $I_1 \cdots I_n$ of maximal integral ideals such that $\ordl{I}=\ordl{I_1}$ and  $\ordr{I}=\ordl{I_n}$. 
\end{thm}

There are multiple ways to determine a factorization of $q$ from the factorization of ideals described in the theorem. In this section we describe an approach that always works, but does not guarantee that the factors are contained in $\mos$. Note that the factorization of $q\mo$ is not sensitive to right multiplication of $q$ by unit $u$ of $\mo$ because $qu\mo = q\mo$. For this reason we will need to treat the unit part of the factorization separately later in the section. 

For the result of this section it is crucial that there are only finitely many conjugacy classes of maximal orders in any quaternion algebra. It is also important that there are finitely many maximal right ideals with a given norm and a given right order. For establishing our results we need to understand how to classify maximal right ideals with a given norm and what is the relation between conjugacy classes of their left and right orders. 

\subsubsection{Classification of maximal right ideals}
Now we discuss the classification of maximal right ideals with a given norm. It is related to the classification of primitive ideals of maximal orders. 

\begin{dfn} A right ideal of  $\mo$ is \emph{\prim{}} if it is not contained in any proper two-sided ideal of $\mo$.
\end{dfn}

Note that above definition is different from one given in \cite{KV}, this one might generalize better to general central simple algebras. Equivalently, we can restrict our attention above only to $J$ being a prime ideal of $\CM$. Our goal is to establish a precise relation between maximal ideals and primitive ideals. We first recall the results describing possible values of the norm of maximal right ideals and establishing relations between maximal right ideals of $\CM$ and prime ideals of $\mo$. 

\begin{thm}[Special case of Theorem 22.15 in \cite{R}]\label{thm:prime-ideals} Each maximal right ideal contains a unique prime ideal $\P$ of $\mo$.
\end{thm}

\begin{thm}[Special case of Theorem 24.13 in \cite{R}]\label{thm:maxid}  Let $I$ be a maximal right ideal of $\mo$.  Then $\nrd(I)$ is prime and equal to $I \cap \zf$.
\end{thm}

In Theorem 24.13 in \cite{R} it is not stated that $I \cap \zf$ is a prime ideal, however this becomes clear after reading the proof of Theorem 24.8. It turns out that the converse is also true: 
\begin{prop} \label{prop:prime-norm}
Let $I$ be a right ideal of $\mo$. If $\nrd(I)$ is prime, then $I$ is maximal.
\end{prop}

\begin{proof}
Let $J$ be a right ideal of $\mo$ containing $I$ and suppose that $\mfp := \nrd(I)$ is prime. We will show that either $J = I$ or $J = \mo$. First note that $\nrd\at{J} \mid \mfp$ and therefore $\nrd\at{J}$ is either equal to $\p$ or $\zf$. The equality $\nrd\at{J}=\zf$ implies $\ord_{\zf}\mo/J = \zf$ which is true if and only if $J = \mo$ (see (24.1) and related discussion in \cite{R}). In the case when $\nrd\at{J} = \p$ we have $J^{-1}I$ is an integral $\mo$-ideal with $\nrd$ equal to $\zf$. We conclude that $J^{-1}I = \mo$ which implies that $I = J$. 
\end{proof}

Some maximal right ideals are also prime. We can distinguish when it is the case based on the norm of the maximal ideal. The following theorem is crucial for this purpose: 

\begin{thm}[Special case of Theorems 25.7, 25.10 in \cite{R}]\label{thm:disc} Let $\CM$ be a maximal order in $\qa$.  For each prime ideal $\mfp$ of $\zf$, define its local index $m(\p)$ in $\CM$ via the equality $\p \mo = \P ^ {m(\p)}$, where $\P$ is the prime ideal of $\mo$ with norm $\p$. 
Then $m(\p)=1$ for all but finitely many $\p$.  The discriminant of $\qa$ is equal to $\prod_{\p} \p^{m(\p)-1}$ where the product is over prime ideals $\p$ of $\zf$. 
\end{thm}

Using the results above we establish the following relation between maximal and primitive ideals: 

\begin{prop} \label{prop:ideal-alternative}
Let $I$ be a right ideal of $\mo$ with prime norm $\mfp$. Then there are two alternatives: 
\begin{enumerate}
	\item $\mfp \mid \disc(\qa)$ and $I$ is prime
	\item $\mfp\nmid \disc(\qa)$ and $I$ is \prim. 
\end{enumerate}
\end{prop}
\begin{proof}
Let $\p = \nrd\at{I}$. By \propos{prime-norm}, the right ideal $I$ is maximal. By \theo{prime-ideals} there exists a prime ideal $\P$ of $\mo$ such that $\P \subset I$. First suppose that $\p$ divides the discriminant of $\qa$. By \theo{disc} we have $\nrd\at{\P}=\p$, therefore $I^{-1}\P$ is an integral $\mo$-ideal with norm equal to $\zf$. We conclude that $I^{-1}\P = \mo$ and $I = \P$, as required. Consider the case when $\p$ does not divide the discriminant of $\qa$. We seek a contradiction by assuming that $I$ is not primitive. This implies $I = \P$, but $\nrd{P}=\p^2$ and $\nrd\at{P}\ne\nrd{I}$.
\end{proof}

Recall the definition of the projective line over a commutative ring $R$: 
\[\pone\at{R} = \set{\at{x,y} : x R + y R = R } / R^\times\]
and let $\at{x:y} = R^\times \cdot (x,y)$ denote the line passing through the point $(x,y)$.   The number of points in $\pone\at{\zf/\a}$ is equal to 
\begin{equation} \label{eq:degree}
\Phi\at{\a} = \prod_{ \p^{e} | \a } \Nm_{\nf/\bQ}\at{\p^{e-1} }\at{\Nm_{\nf/\bQ}\at{\p}+1},
\end{equation}
where the product is over the highest powers $\p^{e}$ of prime ideals dividing $\mfa$ (see Page 27 in \cite{KV}). This will be also very useful for \sec{graph}. The following result classifies the maximal primitive ideals of $\mo$ via an explicit construction:  

\begin{lem}[Special case of Lemma~7.2 from \cite{KV}]\label{lem:prim-enum} Let $\a$ be an ideal of $\zf$ that is co-prime to the discriminant of $\qa$. Then the set of primitive right ideals of $\CM$ with norm $\a$ is in bijection with $\pone\at{\zf/\a}$. Explicitly, given a splitting 
\[
\phi_{\a} : \mo \hookrightarrow \mo \otimes_{\zf} {R_\a} \cong M_2\at{R_\a} \to M_2\at{\zf/ \a},
\]
where $R_\mfa$ is the corresponding localization of $R_F$, 
the bijection is given by 
\begin{align*}
\pone\at{\zf/\a} &\to \set{ I \subset \mo : \nrd\at{I} = \a } \\
\at{x:y} &\mapsto I_{\at{x:y}} = \phi_{\a}^{-1} \UG{x}{0}{y}{0}\mo + \a\mo.
\end{align*}
\end{lem}

We finish this part showing a useful relation between integral and primitive ideals: 
\begin{prop}\label{prop:prim} Each right ideal $I$ of $\mo$ is contained in a unique \prim{} right ideal $J$ of $\mo$ such that $\ordl{I}=\ordl{J}$. There exists an algorithm for finding $J$ (see \fig{prim}). 
\end{prop}

\begin{figure}
\begin{algorithmic}[1]
\Require $I$ -- right integral $\mo$-ideal
\Procedure{PRIMITIVE-IDEAL}{$I$}
\State Factor $\nrd\at{I}$ into a product of prime ideals $\p_1^{k(1)}\ldots\p_m^{k(m)}$ of $\zf$
\State $\CP \gets \set{\P : \P \text{ prime ideal of }\mo \text{ such that } \P \cap \zf = \p_k , k=1,\ldots,m}$
\State $J' \gets I$ \label{alg:prim:base} 
\For{$\P \in \CP$}
\While{ $J' \subset \P$ }
\State $J' \gets J'\P^{-1}$ \label{alg:prim:int} \Comment $J'$ is always integral $\mo$-ideal, $\ordl{J'}=\ordl{I}$
\EndWhile
\EndFor
\State \Return $J$
\EndProcedure
\Ensure $J$ -- \prim{} right $\mo$-ideal such that $I \subset J$ and $\ordl{I}=\ordl{J}$
\end{algorithmic}
\caption{\label{fig:prim}Algorithm for computing primitive ideal }
\end{figure}

\begin{proof}
We first show that the algorithm in \fig{prim} outputs an integral $\ordl{I}$-$\mo$-ideal that contains $I$. Note that during the execution of the algorithm ideal $J'$ always satisfies the mentioned properties. This is the case when $J'$ is initialized with $I$. Before the line~\ref{alg:prim:int} in \fig{prim} is executed we have $J' \subset \P$, which implies that $J' \P^{-1}$ is an integral ideal. Left and right orders of $J' \P^{-1}$ are the same as right and left orders of $J'$ because $\P^{-1}$ is a two-sided ideal of $\mo$. 
\end{proof}

This result implies that we we can always study the factorization of ideals by separately studying factorization of two-sided ideals and factorization of primitive ideals into product of maximal primitive ideals. 

\subsubsection{Conjugacy classes of maximal orders. Adjacency structure.}
The notion of two maximal orders being $\p$-neighbors is closely related to the factorization of integral ideals into a product of maximal integral ideals.  

\begin{dfn}\label{defn:p-neigbor} Maximal orders $\mo'$ and $\mo''$ are called \emph{$\p$-neighbors} if there exists an integral $\mo'$-$\mo''$-ideal $I$ of $\qa$ such that $\nrd\at{I}=\p$. For a set $S$ of prime ideals of $\zf$, we say that $\mo'$ and $\mo''$ are \emph{$S$-neighbors} if they are $\p$-neighbors for some $\p \in S$.  
\end{dfn}

Note that the factorization $I_1 \cdots I_n $ described in \theo{factor} determines a sequence of maximal orders 
\[
\ordl{I_1},\ordr{I_1}=\ordl{I_2},\ldots,\ordr{I_{n-1}}=\ordl{I_n}, \ordr{I_n}
\]
along which all adjacent pairs are $S$-neighbors, 
where $S$ is a set of all prime ideals dividing the norm of $I_1 \ldots I_n$. The maximal orders in any such sequence belong to a finite set of different conjugacy classes (see Problem 2.8 and related discussion in \cite{KV}). The following structure relates the $S$-neighbor relation and conjugacy classes of maximal orders. 

\begin{dfn}\label{defn:adj-desc} Given ideals $\p_1,\ldots,\p_l \subset S$ coprime to $\disc(\qa)$ and representatives $\mo_1,\ldots,\mo_m$ for the conjugacy classes of maximal orders of $\qa$,  the \emph{maximal orders adjacency description} is a sequence $\Adj_{1,1},\ldots,\Adj_{m,l}$. Each $\Adj_{i,j}$ is a sequence $a_{i,j,1}, \ldots, a_{i,j,N_j}$ where $N_j = \Nm_{\nf/\bQ}\at{\p_j}+1$. Each $a_{i,j,k}$ is a pair $\at{s,q}$ where $s$ is an integer and $q$ is an integral element of $\qa$. Each pair $\at{s,q}$ corresponds to a maximal order $q \mo_s q^{-1}$. The sequence $\Adj_{i,j}$ describes all $\p_j$-neighbors of the maximal order $\mo_i$. 
\end{dfn}

For our goals it is important that we can compute the maximal orders adjacency description for a given quaternion algebra $\qa$ and a set $S$ of prime ideals of $\zf$. The algorithms for constructing it using algorithms from \cite{KV} as building blocks is shown on \fig{adj-structure}. It uses the following algorithms  as subroutines: 
\begin{itemize} \label{adj-subroutines}
	\item CONJ-CLASSES-LIST. Returns a list of representatives for the conjugacy classes of maximal orders in a quaternion algebra. Equivalent to finding representatives of one-sided ideal classes.  See also Problem 2.8 and related discussion in \cite{KV}. 
	\item IS-CONJUGATE. Equivalent to testing if a right ideal of quaternion algebra is principal. See Problem  2.7 in \cite{KV}. 
	\item ISOMORPHISM-GENERATOR. Equivalent to finding a generator of a principal right ideal of a quaternion algebra. See Problem 2.4 in \cite{KV}. Through the paper without loss of generality we assume this procedure returns integral quaternions. 
	\item P-NEIGHBORS. Computes the $\p$-neighbors of given maximal order $\mo$. First enumerates all right ideals of $\CM$ with norm $\p$ (using \lemm{prim-enum}) and then computes their left orders. The details are described in \cite{KV}. 
\end{itemize}
Implementation of all these algorithms (or equivalent) is available in MAGMA. 

\begin{figure}[ht]
\begin{algorithmic}[1]
\Statex 
\Require Set $S$ of prime ideal of $\zf$ 
\Procedure{MAX-ORDERS-ADJ}{$S$}
\State Let $\p_1,\ldots,\p_l$ be all elements of $S$
\Statex that does not divide  the  discriminant of $\qa$ \label{line:disc}
\State  $\Adj \gets $ empty maximal orders adjacency description 
\State $\mo_1, \ldots, \mo_m \gets$ \Call{CONJ-CLASSES-LIST}{} 
\ForAll{ $i = 1,\ldots,m$ }
\ForAll{ $j = 1,\ldots,l$ }
	\State $\mo'_1, \ldots, \mo'_{N(\p_j)+1} \gets$ \Call{P-NEIGHBORS}{$\mo_i$,$\p_j$ } \label{line:p-neighbors}
	\ForAll{$k = 1,\ldots,N(\p_j)+1$}
	\State Append \Call{CONJ-CLASS-DESCR}{$\mo'_k$} to $\Adj_{i,j}$
	\EndFor
\EndFor
\EndFor
\State \Return $\Adj$
\EndProcedure 
\Ensure Maximal orders adjacency description $\Adj$ (see \defin{adj-desc}) 
\Statex 

\Require Maximal order of quaternion algebra $\mo$
\Procedure{CONJ-CLASS-DESCR}{$\mo$}
\State $\mo_1, \ldots, \mo_m \gets$ \Call{CONJ-CLASSES-LIST}{}  \Comment See \ref{adj-subroutines}
\ForAll{ $s = 1,\ldots,m$}
\If{\Call{IS-CONJUGATE}{$\mo$,$\mo_s$}}  \Comment See \ref{adj-subroutines}
\State $q \gets \text{ISOMORPHISM-GENERATOR}\at{\mo,\mo_s}$  \Comment See \ref{adj-subroutines}
\State \Return $\at{s,q}$ 
\EndIf
\EndFor
\EndProcedure 
\Ensure $\at{s,q}$, where $s$ is the index of the maximal order $\mo$
\Statex in the sequence $\mo_1,\ldots,\mo_l$ output by 
\Statex \Call{CONJ-CLASSES-LIST}{}, and $\mo=q\mo_s q^{-1}$
\end{algorithmic}
\caption{\label{fig:adj-structure} Algorithm for computing the maximal orders adjacency description (see \defin{adj-desc}) and related sub-routine. }
\end{figure}

We sketch the proof of the correctness of the MAX-ORDERS-ADJ procedure. The crucial point here is that when $\p\nmid \disc(\qa)$ and $\CM$ is a maximal order of $\qa$, there is a one-to-one correspondence between $\p$-neighbors of $\mo$ and primitive maximal right ideals $I$ of $\mo$ with norm $\mfp$. Any such ideal $I$ corresponds to the $\p$-neighbor $\ordl{I}$ of $\mo$. \propos{prim} shows how to find the ideal corresponding to a given $\p$-neighbor of $\mo$. To summarize (keeping the notation from \defin{adj-desc}), it is sufficient to enumerate all primitive right $\mo_i$-ideals with norm $\p_j$ to compute $\Adj_{i,j}$. This is precisely what is done on line~\ref{line:p-neighbors} in \fig{adj-structure}. 

\subsubsection{Reduction of the factorization question to the description of principal two-sided ideals} \label{sec:reduction-to-two-sided} Now we are going to show that, using the maximal orders adjacency description, we can reduce the factorization of an arbitrary $q \in \mos$ to the factorization of elements $q_r$ such that $q_r \mo q^{-1}_r = \mo$. Such elements $q_r$ are generators of two-sided principal ideals $q_r \mo$. 

We are going to show that any $q$ from $\mos$ can be written as $q_1\ldots q_n q_r$ for $q_r$ having properties mentioned above and with $q_1,\ldots,q_n$ belonging to the set $\gone$. We define the set keeping the notation used in \defin{adj-desc}. Let $i_0,q_\mo$ be such that $\mo = q_{\mo} \mo_{i_0} q_{\mo}^{-1}$ and define
\begin{equation}\label{eq:g-one}
\gone = \set{ q_{\mo} q q_{\mo}^{-1} : \at{s,q} \in \Adj_{i,j}, i = 1,\ldots,m, j = 1,\ldots,l }.
\end{equation}
The set $\gone$ can be easily computed given $\Adj$ output by MAX-ORDERS-ADJ and using CONJ-CLASS-DESCR from \fig{adj-structure}. 

The notion of $\p$-adic valuations $v_{\p}$ of number field elements and ideals is crucial for counting the number of terms in factorization of $q$. We also define $v_S = \sum_{\p \in S}v_{\p}$. We use it to state the main result of this part:     

\begin{thm} \label{thm:exact-one} Let $q$ be an element of $\mos$ and let $S_1 \subset S$ contain all prime ideals of $S$ that are co-prime to the discriminant of $\qa$. There exists a factorization $q = q_1 \ldots q_n q_{rem}$ such that $n \le v_{S_1}\at{\nrd\at{q}}$, $q_1,\ldots,q_n \in \gone$ and $q_{rem}\mo$ is a two-sided ideal. More precisely, $n = v_{S_1}(I)$, where $I$ is the unique primitive right ideal of $\mo$ contained in $q\mo$.  There exist algorithms for finding the factorization $q_1 \ldots q_n q_{rem}$ (EXACT-SYNTHESIS-STAGE-1 on \fig{exact-synthesis-stage-1}) and for computing the set $\gone$.  
\end{thm}

\begin{figure}
\begin{algorithmic}[1]
\Require Maximal order $\mo$ of $\qa$, set $S$ of prime ideals of $\zf$,
\Statex an element $q$ of $\mo$ such that $\nrd\at{q}\zf$ factors into ideals from $S$ 

\Procedure{EXACT-SYNTHESIS-STAGE-1}{$\mo$,$S$,$q$}

\State $\at{s_0,q_0} \gets $ \Call{CONJ-CLASS-DESCR}{$\mo$}, $s \gets s_0$
\State $I \gets q^{-1}_0 \text{PRIMITIVE-IDEAL}\at{q\mo}  q_0 $ \Comment $\ordr{I} = \mo_{s}$
\State $C' \gets $ empty sequence of elements of $\qa$
\While{ $\ordr{I} \ne \ordl{I} $}
\State $\at{s,I,q_c} \gets$ \Call{REDUCED-IDEAL}{$s$,$I$,$S$} \Comment $\ordr{I} = \mo_{s}$
\State Add $q_c$ to the end of $C'$
\EndWhile
\State Let $C' = q'_1,\ldots, q'_n $
\State Let $C = q_0 q'_1 q_0^{-1} ,\ldots, q_0 q'_n q_0^{-1} = q_1, \ldots, q_n $ \Comment Sequence of elements of $\gone$
\State $q_{rem} \gets  q_n^{-1} \ldots q_1^{-1} q $ \Comment Now $q_{rem} \mo q_{rem}^{-1} = \mo$
\State \Return $\at{C,q_r}$
\EndProcedure 
\Ensure Sequence $C = q_1,\ldots,q_n$ and quaternion $q_r$ such that  $q_1 \ldots q_n q_r = q $ where 
\Statex $q_i$ are from $\gone$ and $q_r\mo$ is a two-sided $\mo$-ideal
\Statex 
\Require Primitive right $\mo$-ideal $I$ such that its $\nrd$ factors into elements of $S$,
\Statex set $S$ of prime ideals of $\zf$
\Statex where $s$ is an index of $\mo$ in the sequence output 
\Statex by CONJ-CLASSES-LIST
\Procedure{REDUCED-IDEAL}{$s$,$I$,$S$}
\State Let $S_1 = \set{\p_1,\ldots,\p_l}$ be all elements of $S$
\Statex that does not divide  the  discriminant of $\qa$ 
\ForAll{ $k = 1,\ldots,l$}
\State $N \gets $All right $\mo_{s}$-ideals with norm $\p_l$ \Comment Precomputed for all $s,l$
\ForAll{ $I' \in N$ }
\If{$ I \subset I' $}
\State $I \gets I \at{I'}^{-1}$
\State $\at{s,q_c}\gets $ \Call{CONJ-CLASS-DESCR}{$\ordr{I}$}
\State \Return $\at{s, q^{-1}_c I q_c ,q_c}$ 
\EndIf
\EndFor 
\EndFor
\State \Return $\at{0,I,0}$ \Comment We prove that this point is never reached
\EndProcedure
\Ensure $\at{s,I_{r},q_c}$, where $I'$ is a right $\mo$-ideal such that $s$ 
\Statex is an index of $\mo$ in the sequence output by CONJ-CLASSES-LIST,
\Statex ideal $q_c I_{r} q^{-1}_c$ is a right ideal that contains $I$, has the same
\Statex left order as $I$ and $v_{S_1}\at{\nrd\at{I_{r}}} = v_{S_1}\at{\nrd\at{I}}-1 \ge 0$. 
\Statex If $v_{S_1}\at{\nrd\at{I'}} >0 $, the ideal $I'$ is primitive, $I'$ is an order otherwise.  
\end{algorithmic}
\caption{\label{fig:exact-synthesis-stage-1} Reduction of factorization of $q$ from $\mos$ to factorization of elements that generate two-sided $\mo$-ideals}
\end{figure}

\begin{proof} To prove the theorem we will prove the correctness of EXACT-SYNTHESIS-STAGE-1 on \fig{exact-synthesis-stage-1}. We first prove that EXACT-SYNTHESIS-STAGE-1 is correct under the assumption that REDUCED-IDEAL is correct and then show the correctness of REDUCED-IDEAL.

First note that the {\bf while} loop of the procedure terminates in $n$ steps. This is the case because REDUCED-IDEAL reduces $v_{S_1}\at{\nrd\at{I}}$ by one. Recall that $\mo_1,\ldots,\mo_m$ are representatives of the conjugacy classes of maximal orders of $\qa$. Each time before and after the execution of the {\bf while} loop body the following is true: 
\[
 \ordr{I}=\mo_s, \ordl{I} = q' \mo \at{q'}^{-1}, q' = \at{q'_0 q'_1 \ldots q'_k }^{-1}q,\text{ where } C'=q'_1,\ldots,q'_k 
\]
After the execution of the while loop we have $\ordl{I} = \ordr{I}$, and they must be equal to $\mo_{s_0}$ because $\ordl{I}$ is conjugate of $\mo_{s_0}$. We get the following for $C'=q'_1,\ldots,q'_n$:
\[
 \mo_{s_0} = q' \mo \at{q'}^{-1}, q' = \at{q_0 q_1 \ldots q_n }^{-1}q, \,\,\,\,\,\mo = q_0 q' \mo \at{q'}^{-1} q^{-1}_0
\]
Next we rewrite $q_0 q'$ as $\at{ q_0 q'_n q^{-1}_0 }^{-1} \ldots \at{ q_0 q'_1 q^{-1}_0 }^{-1} q $ and see that this element generates a two-sided $\mo$-ideal.  Also note that all elements $q_0 q'_k q^{-1}$ are from $\gone$ by construction. From this we conclude that the output of the algorithm is correct and is a factorization of $q$ described in the statement of the theorem.  

Let us now show that procedure REDUCED-IDEAL on \fig{exact-synthesis-stage-1} works as specified. Let $I$ be a primitive $\mo_s$-ideal input to REDUCED-IDEAL. According to \theo{factor} there exists a factorization of $I$ into a product $I_1 \cdots I_n$ of maximal ideals. Each of maximal ideals $I_k$ must be primitive and have norm from $S_1$, because the converse would contradict to $I$ being a primitive ideal. Moreover, $I_n$ must be from the set $N$. Therefore there exist at least one element $I'$ of $N$ such that $I \subset I'$. Note that  $I \subset I'$ implies that there exits an integral ideal $J$ such that $I = J I'$. We see that $J = I (I')^{-1}$. Equality $v_{S_1}\at{\nrd\at{I'}}=1$ and multiplicativity of $\nrd$ implies that  $v_{S_1}\at{\nrd\at{J}} = v_{S_1}\at{\nrd\at{I}} -1$. The norm of $J$ factors into elements of $S_1$, therefore $J$ can be an order only in the case when  $v_{S_1}\at{\nrd\at{J}}=0$. All other claimed properties of the output of REDUCED-IDEAL can be easily checked by inspection.  
\end{proof}

Let us now consider the restriction on the $\nrd$ of $q_{rem}$. Condition $q_{rem} = \at{q_1 \ldots q_n}^{-1}q$ implies that $\nrd\at{q_{rem}}\zf$ factors into ideals from the following set: 
\begin{equation} \label{eq:s-ext}
S_{ext} = \set{ \p : \p \text{ divides } \nrd\at{q}\zf, q \in \gone } \cup S
\end{equation}

Note that there is some freedom in choosing elements $\gone$.  Recall that any element of $\gone$ has form $q_0 q q^{-1}_0$. The quaternion $q$ comes from some pair $(s,q)$ of $\Adj_{i,j}$ ($\Adj$ is the adjacency structure of maximal orders). The quaternion $q_0$ is chosen such that $q_0 \mo q^{-1}_0=\mo_k$. Maximal order $\mo_k$ is an element of the list for $\mo_1,\ldots,\mo_m$ of representatives of conjugacy classes of $\qa$ fixed throughout this section. Note that according to the definition of $\Adj$ we can replace $(s,q)$ with $(s,q')$ if $q \mo_s q^{-1} = q' \mo_s q'^{-1}$. In other words, pairs $(s,q)$ and $(s,q')$ are equivalent if $q^{-1} q'$ generates a principal two-sided ideal of $\mo_s$. 

Using discussed freedom of choosing elements of $\gone$ one might try to make the set $S_{ext}$ as small as possible. Consider first the situation when two-sided ideal class group is trivial. This means that we can always choose pair $(s,q)$ such that ideal $q\mo_s$ is primitive and has norm coprime to the discriminant of $\qa$. Indeed if $q \mo_s$ is not primitive we can always write it as $I J$ for $J$ being a two-sided $\mo_s$-ideal and $I$ being primitive ideal. Under our assumptions, $J$ is a primitive ideal, therefore $I$ is also primitive. We then replace $\at{s,q}$ with $\at{s,q'}$ where $q'\mo_s = I$.

In general, the two-sided ideal class group is a finite Abelian group. In this case, for each pair $(s,q)$ from $\Adj_{i,j}$ one might compute a primitive $\mo_s$-ideal $I_q$ such that $q\mo_s = I_q J_q$ and try to optimize over possible choices of two-sided ideals $J'_q$ with the same image in the two-sided class group as $J_q$ and such that the norm of $J'_q$: (a) has small number of factors and (b) is co-prime to discriminant of $\qa$. We will see later in this section why it is useful to ensure that $\nrd{q}$ in pairs $\at{s,q}$ is coprime to the the discriminant of $\qa$.

\subsubsection{Factorization of generators of two-sided ideals} In this section we will study the factorization of two-sided ideals with norm supported in $S$. Out goal is to prove the following result: 
\begin{prop} \label{prop:two-sided-dec} Let $q$ generate a two-sided ideal of $\CM$ whose norm is supported in $S$. Then $q$ can be written as $q_1 \ldots q_n u \alpha$, where $q_1,\ldots,q_n$ are from the finite set 
$\gtwo$,
$\alpha \in \nf^\times$ and $u \in\mo^\times$. There exist algorithms for finding the factorization and for computing the set $\gtwo$. 
\end{prop}
\begin{proof} 
We begin by writing $q  = q_0 \al$, where $\al \in F^\times$ is such that $q_0$ is an integral element of $\qa$ with norm supported in $S$ (i.e.\ an $S$-unit).
Let $\p_1,\ldots,\p_l$ be the elements of $S$ coprime to the discriminant of $\qa$ and let $\p_{l+1},\ldots,\p_{l+n}$ be the ones that divide the discriminant of $\qa$. For $k=l+1,\ldots,l+n$, let $P_k$ be the prime ideals of $\mo$ such that $\P^2_k = \p_k\CM$.  We can write the ideal $q_0 \mo $ as
\[
 \p^{c(1)}_1 \cdots \p^{c(l)}_l \P^{c(l+1)}_{l+1} \cdots \P^{c(l+n)}_{l+n}. 
\] 
 For $k=1,\ldots,l$, let $h(k)$ be the smallest positive integer such that $\p_k^{h(k)}$ is principal and such that the ideal  $\P_k^{h(k)}=\tilde{q}_k\mo$ is principal. Then we can define:
\begin{align} \label{eq:two-sided-gen}
\gtwo = & \set{ \tilde{q}_k : k = l+1, \ldots, l + n }  \cup &  \\ 
& \set{ q : q\mo = \p^{c(1)}_1 \ldots \p^{c(l)}_l \P^{c(l+1)}_{l+1} \ldots \P^{c(l+n)}_{l+n} , c(k) \le h(k) }. & \nonumber 
\end{align}
It is easy to see that any two-sided integral ideal can be written as a product of elements $q\mo$ times some element of $\zf$. If two quaternions generate the same two-sided ideal of $\mo$  they can only differ by unit of $\mo$. This concludes the proof.    
\end{proof}

When the two-sided ideal class group of $\qa$ is trivial, the above proposition can be simplified as follows: 

\begin{prop} Let $\P_1,\ldots,\P_n$ be prime ideals with norm dividing the discriminant of $\qa$ and from $S$. Let $q_1,\ldots,q_n$ be generators of $\P_1,\ldots ,\P_n$. Any element that generates a two-sided ideal with norm supported in $S$ can be written as a product $\alpha q_1^{c(1)}\ldots q_n^{c(n)}u$, where $c(k) \in {0,1}$, $\alpha \in \nf^\times$ and $u \in \mo^\times$. 
\end{prop}

Based on the discussion in the end of \sec{reduction-to-two-sided} we can also make $S_{ext} \subset S$.  This simplified situation will frequently appear in our applications. This leads to the simple algorithm TWO-SIDED-DECOMPOSE, which decomposes generators of two-sided ideals and is shown in \fig{exact-synthesis}. 

\subsubsection{Unit group elements factorization\label{sec:unit-groups}}  To the best of our knowledge the algorithmic aspect of this question is well-studied in two following situations: 
\begin{enumerate}
	\item $F$ is a totally real number field, and $\qa$ is a totally definite quaternion algebra. 
	\item $F$ is a totally real number field, and $\qa$ is split in exactly one real place.  
\end{enumerate}
In both cases there are algorithms for computing the unit group of a maximal order, for finding its generators and for solving the word problem in terms of generators \cite{KV,V}. Implementations of related algorithms are available in MAGMA. In our pseudo-code we refer to procedures for decomposing units as UNIT-DECOMPOSE. 

\subsubsection{Summary} Now we summarize results obtained above and prove \theo{factor-1}. We also combined algorithms discussed above into the algorithm for finding a factorization. In \fig{exact-synthesis-two} we show the pseudo-code for the exact synthesis algorithm in the case when two-sided ideal class group is trivial. It is not difficult to obtain an algorithm for more general cases based on results of this section. 

\begin{figure}
\begin{algorithmic}[1]
\Require Maximal order $\mo$ of $\qa$, set $S$ of prime ideals of $\zf$,
\Statex an element $q$ of $\mo$ such that $\nrd\at{q}\zf$ factors into ideals from $S$ 
\Statex (assumes that two side ideal class group of $\qa$ is trivial)

\Procedure{EXACT-SYNTHESIS-2}{$\mo$,$S$,$q$}
\State $\at{C,q} \gets$ \Call{EXACT-SYNTHESIS-STAGE-1}{$\mo$,$S$,$q$}
\State Add \Call{TWO-SIDED-DECOMPOSE}{$q$,$\mo$,$S$} to the end of $C$  \Comment \fig{exact-synthesis}
\State \Return C
\EndProcedure 
\Ensure Outputs  $q_1 \ldots q_m u_1 \ldots u_n = q $ where 
\Statex $q_i$ are from $\gext$ and $u_j$ are generators of the unit group of $\mo$ 
\end{algorithmic}
\caption{\label{fig:exact-synthesis-two} Algorithm for finding the factorization of elements of $\mos$ described in \theo{factor-1} when the case of trivial two-sided ideal class group. }
\end{figure}

\begin{proof}[Proof of Theorem~\ref{thm:factor-1}] To prove the theorem, we combine results of \theo{exact-one} and \propos{two-sided-dec}. For the decomposition of units and finding the generators of the unit group we use results discussed in \sec{unit-groups}. Recalling the definition of $S_{ext}$ (equation \eq{s-ext}) and $\gtwo$ (equation \eq{two-sided-gen}) we conclude that $\gext = \gone \cup \gtwoext$. An algorithm for finding the factorization is shown on \fig{exact-synthesis-two} in the special case when two-sided ideal class group is one. Algorithms for finding $S_{ext}$, $\gone$ and $\gtwoext$ were discussed above. 
\end{proof}

\subsection{  Factorization of elements of $\mos$ into a finite subset of $\mos$   }  \label{sec:graph}


In this subsection we use an approach similar to \sec{factor-1}, but describe another way of coming back from the factorization of an ideal $q\mo$ into a product of maximal ideals $I_1 \cdots I_n$ of $\qa$ (described in \theo{factor}) to the factorization of $q$ into a product of finitely many elements of $\mos$. The basic idea is to group consecutive maximal ideals in $I_1 \cdots I_n$ into products $\at{I_{k(0)} \cdots I_{k(1)}} \cdots \at{I_{k(m-1)} \cdots I_{k(m)}}$ such that each block $\at{I_{k(j-1)} \cdots I_{k(j)}}$ is a principal ideal. This approach finds a factorization of $q$ into finite set of elements of $\mos$ under some technical assumptions. To formulate them we introduce the notion of the {\em ideal principality graph}. 

\subsubsection{Ideal principality graph} We first introduce the definitions required to state the main result of this subsection. 

\begin{dfn} The ideal principality graph $\gs$ is an oriented graph whose vertices are maximal orders of the quaternion algebra $\qa$ and is defined by the following rules: 
\begin{itemize}
\item $\mo$ is a vertex of $\gs$ (called a root of $\gs$). 
\item Let $\mo'$ be $\mo$ or a vertex of $\gs$ that it is not a conjugate of $\mo$ and let $\mo''$ be a maximal order such that there primitive integral $\mo''-\mo'$-ideal $I$ with norm from $S$, then $\mo''$ is a vertex of $\gs$ connected to $\mo'$ by arc $\at{\mo',\mo''}$. All such maximal orders $\mo''$ are vertices of $\gs$.
\item Vertices of $\gs$ that are conjugate of $\mo$ and not equal to $\mo$ are called leaves of $\gs$. 
\end{itemize}
\end{dfn}

Note that $\gs$ is  connected because by definition every vertex of $\gs$ is connected to $\mo$, all leaves have outdegree $0$, outdegree of all other vertices is the same. Leaves of $\gs$ are closely related to all possible principle ideals we can get when considering blocks $\at{I_{k(j-1)} \cdots I_{k(j)}}$ mentioned above. 

\begin{dfn} Let $\mo'$ be a leaf of $\gs$ and $q(\mo')$ be an element of $\mos$ such that $\mo' = q \mo q^{-1}$, then the set of leaf generators is
\[
\gens = \set{ q\at{\mo'} : \mo' \text{ is a leaf of } \gs }.
\]
\end{dfn}

The following theorem is the main result of this section:
\begin{thm} \label{thm:factor-2} Suppose that: 
\begin{enumerate}
	\item graph $\gs$ is a finite graph,
	\item each element of $\gens$ can be chosen such that its $\nrd$ factors into elements of $S$,
\end{enumerate}
then any element of $\mos$ can be written as a product $q_1 \ldots q_n u \alpha$, where $q_1,\ldots,q_n \in \gens \cup \gtwo$, $u \in \CM^\times$ and $\alpha \in \nf^\times$. 

There are algorithms for checking that $(1)$ and $(2)$  are true, for computing $\gens \cup \gtwo$ and for finding a factorization $q_1 \ldots q_n u \alpha$ when the conditions $(1)$ and $(2)$ are true. 

If one of the following conditions holds, the unit $u$ can be written as a product of finite number elements of $\ugens$: 
\begin{enumerate}
	\item $F$ is a totally real number field, and $\qa$ is a totally definite quaternion algebra. 
	\item $F$ is a totally real number field, and $\qa$ is split in exactly one real place.  
\end{enumerate}
There are algorithms for finding $\ugens$ (see \sec{unit-groups}) in both cases mentioned above.
\end{thm}

For the algorithmic solution to $(2)$ in the theorem above see \rema{bad-primes}. For the algorithm for finding a factorization, see procedure EXACT-SYNTHESIS-1 on \fig{exact-synthesis}. For the algorithm for deciding if $\gs$ is finite graph, see procedure SPANNING-TREE-SIZE on \fig{tree-size}. For the algorithm that finds $\gens$, see procedure FIND-S-GENERATORS on \fig{gs-alg}. See also equation \eq{two-sided-gen} for the definition of $\gtwo$ and related discussion about computing it. 

Now we show how to use the ideal principality graph to reduce the problem of factoring elements of $\mos$ to factoring elements that generate two-sided ideals of $\mo$. We will use the following integral complexity measure to state the result:

\begin{dfn} The integral complexity measure of a primitive ideal $I$ is defined as $ \cmi{I} = v_{S}\at{\nrd\at{I}} $. For $q$ being an elements of $\mos$ and $I_q$ being the unique primitive $\mo$-ideal that contains $q\mo$, the $\cmi{q}=\cmi{I_q}$. For two maximal orders $\mo',\mo''$ we define the distance $\cmi{\mo',\mo''}$ to be a complexity measure $\cmi{I}$ of the unique primitive ideal $I$ that contains ideal $\mo'\mo''$. 
\end{dfn}

Note that for $q$ corresponding to a leaf of $\gs$ its integral complexity measure $\cmi{q}$ is equal to the distance from the leaf to the root of $\gs$.

\begin{thm}\label{thm:mos-factor}
If $\gs$ is a finite graph then, for any element $q$ of $\mos$ there exist elements $q_1,\ldots,q_m$ of $\gens$, and a quaternion $q_0$ such that $\mo = q_0 \mo q_0^{-1}$, $q= q_1 \ldots q_m q_0$ and $\cmi{q} = \sum_{k=1}^{m}\cmi{q_k}$.
\end{thm}

\begin{proof}
By \propos{prim} there exists a unique primitive right $\mo$-ideal $I$ such that $q\mo = I J_0$, for $J_0$ being a two-sided $\mo$-ideal. By \theo{factor}, the ideal $I$ can be written as a proper product of maximal ideals $I_n \cdots I_1$. Let $\mo_0 = \mo$ and define $\mo_k = \ordl{I_k}$ for $k=1,\ldots,n$. The maximal orders $\mo_k$ must be distinct, as otherwise the ideal $I$ would be contained in a nontrivial two-sided ideal, contradicting its primitivity. Note that $\mo_n = q \mo q^{-1}$, in other words $\mo_n$ is isomorphic to $\mo$. Let $k(0) = 0$ and choose $k(1)$ such $\mo_{k(1)}$ is the first order in the sequence $\mo_1,\ldots,\mo_n$ isomorphic to $\mo$. Then choose $k(2)$ such that  $\mo_{k(2)}$ is the second order isomorphic to $\mo$ in the sequence, etc.\ Build the sequence $k(0),k(1),\ldots,k(m)$ as above. Note that $k(m)=n$ and the length of the sequence is at most two, because $\mo_n$ is isomorphic to $\mo$. 

Now we establish a correspondence between sequences of maximal orders 
\[
\mo_{k(j-1)} , \ldots,  \mo_{k(j)}, j=1,\ldots,m
\]
and paths in the graph $\gs$ from the root to the one of the leaves. We first show that the sequence $\mo_0,\ldots,\mo_{k(1)}$ corresponds to the path of $\gs$. We use induction by the number of maximal orders in the sequence. For the base case, we show that $\at{\mo_0,\mo_1}$ is an edge of $\gs$ and $\mo_1$ is a vertex of $\gs$. It is sufficient to show that $\nrd\at{I_1}$ is in $S$ because any maximal ideal is primitive. By \theo{maxid}, $\nrd\at{I_1}$ must be a prime ideal of $\zf$. By multiplicativity of the norm of ideals, we have that $\nrd\at{q\mo} \subset \nrd\at{I} \subset \nrd\at{I_k}$ for any $k$, so therefore $\nrd\at{I_k}$ must be an element of $S$. 
If for some $k \le k(1)$, maximal order $\mo_{k-1}$ is a vertex of $\gs$, then $\mo_{k}$ is also a vertex of $\gs$ and $(\mo_{k-1},\mo_k)$ is an arc of $\gs$. This is because $I_k$ is a primitive ideal with norm from $S$ and $\mo_{k-1}$ is not a conjugate of $\mo$ (by construction of the sequence $\mo_0,\ldots,\mo_{k(1)}$). Finally we note that $\mo_{k(1)}$ is a leaf of $\gs$, therefore $\mo_{k(0)} = q^{-1}_1 \mo_{k(1)} q_1$ for $q_1$ from $\gens$. Now we apply the isomorphism $x \mapsto q^{-1}_1 x q_1$ to the sequence of orders $\mo_{k(1)},\ldots,\mo_{k(2)}$ and ideals $I_{k(1)+1},\ldots,I_{k(2)}$ and get a new sequence of orders and ideals. The argument above now applies to the new sequence of orders as it starts with $\mo$. We find $q_2$ such that 
\[
   \mo_{k(0)} = q^{-1}_2 q^{-1}_1 \mo_{k(2)} q_1 q_2. 
\] 
We proceed in a similar way and find a sequence $q_1,\ldots,q_m$ of elements of $\gens$. 

Let us now decompose the ideal $I$ into principal ideals generated by $q_k$ and two-sided $\mo$-ideals. Consider a family of two-sided ideals $J_k$ such that $q_k \mo = I\at{q_k} J_k$, where $I\at{q_k}$ is the unique primitive ideal containing $q_k \mo$. We show that 
\begin{equation} 
 q \mo = q_1 \cdots q_m \mo \at{ J_0 \prod_{k=1}^{m} J^{-1}_k } \label{eq:ideal-dec}
\end{equation}
where $q_1 \cdots q_m \mo$ is a product of two-sided principal ideals generated by $q_m$. Consider the ideal $I_{k(1)}\cdots I_1$. Is is a primitive a ideal contained in $q_1 \mo $. It is unique and equal to $q_1 \mo J^{-1}_1$. Now consider $q^{-1}_1 I_{k(2)} q_1\cdots q^{-1}_1 I_{k(1)+1} q_1$. It is the unique primitive ideal contained in $q_2 \mo$, therefore  
\[
q_2 \mo J_2^{-1} = q^{-1} I_{k(2)} q_1\cdots q^{-1}_1 I_{k(1)+1} q_1 
\]
Next we express $I_{k(2)} \cdots I_{k(1)+1}$ as  
$
  \at{q_1 q_2 q^{-1}_1} \mo_{k(1)} q_1  J_2^{-1} q^{-1}_1
$. Therefore the product $I_{k(2)} \cdots I_{1}$ is equal to
\[
  \at{q_1 q_2 q^{-1}_1  \mo_{k(1)} } \at{q_1  J_2^{-1} q^{-1}_1} \at{ q_1 \mo_{k(0)} }= \at{q_1 q_2 q^{-1}_1 \mo_{k(1)} } \at{ q_1 \mo_{k(0)} } J_2^{-1} J_1^{-1}. 
\]
Next we note that $\at{q_1 q_2 q^{-1}_1 \mo_{k(1)} } \at{ q_1 \mo_{k(0)} } = q_1 q_2 \mo_{k(0)}$. Using the same argument repeatedly we get equality \eq{ideal-dec}. Multiplicativity of the norm of ideals implies that $\cmi{q} = \sum_{k=1}^{m}\cmi{q_k}$. We set $q_0 = q q_1^{-1} \cdots q_m^{-1}$ and note that $q_0$ generates two-sided $\mo$-ideal $ J_0 \prod_{k=1}^{m} J^{-1}_k$ which is equivalent to $q_0 \mo q^{-1}_0 = \mo$.

\end{proof}

\begin{rem}\label{rem:bad-primes} There is some freedom in choosing elements of $\gens$. Also, the definition of $\gens$ does not guarantee that elements of $\gens$ have norm from $S$. We will use ideas similar to ones we used to refine the maximal orders adjacency structure. Consider an element $q \in \gens$ and let $I_q$ be the unique primitive ideal that contains $q\mo$. We can always write $q\mo = I_q J_q$, where $J_q$ is a two-sided ideal of $\mo$. If the two-sided ideal class group of $\qa$ is trivial, then the ideal $I_q$ is principal and we can choose $q$ to be its generator. This will ensure that $q$ has norm from $S$. 

Let now consider how to choose $q$ in the case when two-sided ideal class group is non-trivial. It is a finite Abelian group and therefore isomorphic to the group $\bigoplus_{i=1}^D \z/d_i\z $ for some positive integers $d_1,\ldots,d_D$. Let $x=\at{x_1, \ldots, x_D}$ be an element of ideal class group corresponding to $J_q$ and let $c_i$ for $i=1,\ldots,|S|$ be elements of two-sided ideal class group corresponding to prime two-sided $\mo$-ideals $\P_i$ with norm in $S$. We would like to find prime ideal $J'_q = \prod_{i=1}^{|S|}\P_i^{a_i}$ such that the class group element corresponding to it is equal to $x$. To do this it is sufficient to solve the following system of congruences 
\[
  \sum_{i = 1}^{|S|} c_{i,j} a_i = x_j \at{\mathrm{mod}\,d_j} , j = 1,\ldots,D   
\]
for integers $a_i$. If the system is solvable for all $q$ from $\gens$ we can choose all $\gens$ ho have norm from $S$. This leads to a straightforward algorithm for deciding if all elements of $\gens$ can be chosen to have norm from $S$. 
\end{rem}

\begin{cor} Let $\gs$ be a finite graph. For any element $q$ of $\mos$ such that $\mu\at{q} > 0$ there always exists $q_1$ from $\gens$ such that $\cmi{q_1^{-1} q} < \cmi{q}$. 
\end{cor}

The Corollary above immediately implies the correctness of procedure EXACT-SYNTHESIS-1 on \fig{exact-synthesis} for finding a factorization described in \theo{mos-factor}. Procedure EXACT-SYNTHESIS-1 is restricted to the case when two-sided ideal class group of $\qa$ is trivial. It is not difficult to extend it to the general case, but would make exposition more tedious. For more details related to this see discussion related to TWO-SIDED-DECOMPOSE procedure. 

\subsubsection{How unique is the factorization of elements of $\mos$? Spanning trees of $\gs$. } 

The rest of this section is devoted to the theory needed to prove the correctness of procedures SPANNING-TREE-SIZE on \fig{tree-size} and FIND-S-GENERATORS on \fig{gs-alg}. Our ultimate goal is to compute the description of $\gs$ sufficient for finding its leaves. In other words, it is sufficient for us to find a subgraph of $\gs$ that contains at least one path connecting the root of $\gs$ and each of the the leaves of $\gs$. Each path connecting the leaf $\mo'$ of $\gs$ and its root $\mo$ correspond to a factorization of a unique primitive ideal $I$ that contains ideal $\mo'\mo$. The number of such paths is equal to the number of different factorizations of $I$ into primitive maximal ideals. We will show that we can build a required subgraph of $\gs$ by restricting the class of factorizations of $I$. Moreover, we will show that this subgraph is a spanning tree of $\gs$. 

First we show the connection between the notion of $S$-neighbors and the distance between orders. We call two maximal orders $\mo'$ and $\mo''$ \emph{$S$-connected} if the unique primitive ideal $I$ containing $\mo'\mo''$ has $\nrd$ that factors into elements of $S$. Equivalently, they are $S$-connected if there exists a sequence of maximal orders $\mo' = \mo'_1, \ldots, \mo'_n= \mo''$ such that for each $i$, $\mo'_i$ and $\mo'_{i+1}$ are $S$-neighbors.  

\begin{prop} \label{prop:connect} Let $\mo'$ be $S$-connected to $\mo$ and let $\mo''$ be an $S$-neighbor of $\mo'$. Then either $\cmi{\mo,\mo''} = \cmi{\mo,\mo'} + 1$, or $\cmi{\mo,\mo''} = \cmi{\mo,\mo'} - 1$. 
\end{prop}
\begin{proof} Let $I_0$ be the primitive ideal contained in the ideal $\mo \mo'$ and $I_1$ be the primitive ideal contained in $\mo' \mo''$. If $I_0 I_1$ is primitive we get $\cmi{\mo,\mo''} = \cmi{ I_0 } + \cmi{I_1} = \cmi{\mo,\mo'} + 1$. Let us now consider the case when $I_0 I_1 $ is not primitive and equal to $I J$ where $I$ is a primitive ideal and $J$ is a two-sided ideal. Let $\P$ be the $\ordr{I_1}$ prime ideal that contains $I_0 I_1$. Note that $\nrd\at{I_1}$ must divide $\nrd\at{\P}$. Suppose $\nrd\at{I_1}$ does not divide $\nrd\at{\P}$. Then there is a prime ideal $\P'$ of $\ordr{I_1}$ such that $I_1^{c} I_1  = \P'$ and $I_0 I_1 \subset \P'$. This implies that $I_0 I_1 \subset \P' \cap \P = \P \P'$. This implies that $I_0 \subset \P$ which is a contradiction to $I_0$ being primitive ideal. Because $I_1$ is a maximal ideal such that its $\nrd$ divides $\nrd\at{\P}$ we can write $\P = I_1^{c} I_1$ for $I_1^{c}$ being an integral ideal. We see that $I_0 = I (J \P^{-1}) I_1^{c}$. Using that $I_0$ is primitive we conclude that two-sided ideal $(J \P^{-1})$ is trivial and $I_0 = I I_1^{c}$. By \propos{ideal-alternative} $\nrd{I_c}$ does not divide the discriminant of the algebra $\qa$ and $\cmi{I_c}=1$. We conclude that $\cmi{\mo,\mo''} = \cmi{\mo,\mo'} - 1$. 
\end{proof}

The following result shows when a primitive ideal has a unique factorization. In particular, it implies that if $S$ contains only one prime that is coprime to the discriminant of $\qa$ then the graph $\gs$ is a tree.  

\begin{prop} Suppose that $\p$ does not divide the discriminant of $\qa$. Let $I$ be a primitive right $\mo$-ideal with $\nrd(I)=\p^n$, then the factorization of $I$ described in \theo{factor} into a product of maximal ideals is unique.  
\end{prop}
\begin{proof} We will show the result using an induction on $n$. For $n=1$ it is trivial. Assume that it is true of $k$ and suppose that there are two factorizations $I_{k+1} I_k \cdots  I_{1} $ and $I'_{k+1} I'_k \cdots  I'_{1} $ of $I$. First suppose that $I_{k+1} = I'_{k+1}$. We note that $\ordl{I_k}=\ordl{I'_k}$ and $I_k \cdots I_1 = I'_k \cdots I'_1$ as a unique primitive ideal connecting $\ordl{I_k}$ and $\mo$. Factorizations $I'_1 \cdots I'_k$ and $I_1 \cdots I_k $ must be the same by induction hypothesis. Let us now consider the case $I_{k+1} \ne I'_{k+1}$. We are seeking a contradiction by counting the total number of ideals that has factorization length $k+1$. Every primitive right $\mo$-ideal with the factorization length $k+1$ can be written as $J_{(1)}J_{(k)}$ where $J_{(k)}$ is a primitive right $\mo$-ideal with factorization length $k$ and $J_{(1)}$ is a maximal right ideal with norm $\p$. Therefore $\ordl{J_{(1)}}$ is an $S$ neighbor of $\ordl{J_{(k)}}$. Each maximal order has $(N\at{\p}+1)$ $\p$ neighbours according to \lemm{prim-enum}. According to \propos{connect} they can be distance $k+1$ or $k-1$ from $\mo$. In our case there is only one  $\p$ neighbour that is distance $k-1$ from $\mo$. If we had more than one we would get two different factorizations of length $k$ of the same ideal, which contradicts induction hypothesis. We conclude that each maximal order distance $k$ from $\mo$ with $\nrd$ $\p^k$ must have $N\at{\p}$ $\p$ neighbours of distance $k+1$ from $\mo$. Therefore there is at most $N\at{\p^{k}}(N\at{\p}+1)$ such maximal orders. On the other hand by equation \eq{degree} we know that there is precisely that many of them. This implies that $I_{k+1} \ne I'_{k+1}$ is not possible, because the inequality would imply that $\ordl{I}$ is a $\p$ neighbor of two different maximal orders $\ordr{I_{k+1}}$ and $\ordr{I'_{k+1}}$. 
\end{proof}

Out next goal is to show that ordering ideals in the factorization by their norms makes the factorization unique. Next we note that we can reorder maximal ideals in the factorization, such that their norms are in any preassigned order. This follows from the following: 

\begin{thm}[Special case of Theorem 22.28 in \cite{R}] \label{thm:ordering} Let $I$ be a left ideal of $\mo$ and let $\set{S_1,\ldots,S_n}$ be the composition factors of the $\mo$-module $\mo/I$ arranged in any preassigned order. Then there is a factorization of $I$ as described in \theo{factor} such that the factor modules
\[
\ordl{I_1} / I_1,\, I_1 / I_1 I_2,\, \ldots,\, I_1 \cdots I_{n-1} / I_1 \cdots I_{n} 
\] are precisely $S_1,\dotsc,S_n$ in that order. 
\end{thm}
Note that according to the proof of Theorem 22.24 of \cite{R}, the composition factors $S_k$ uniquely define prime ideals $\P_k$ contained in maximal ideals $I_k$. The immediate corollary of the theorem above is: 
\begin{cor} \label{cor:ordering} Let $I$ be a right ideal of $\mo$ and let $\p_1,\ldots,\p_n$ be any sequence of prime ideals of $\zf$ (not necessarily distinct) such that $\nrd\at{I} = \prod_{k=1}^n \p_k $. Then there exists a factorization $I_1 \cdots I_n$ of $I$ into maximal ideals of $\qa$ such that $\nrd\at{I_k} = \p_k$. 
\end{cor}

We use the corollary above to define class of factorizations that is unique. 

\begin{prop} \label{prop:ordered} Let $\p_1,\ldots,\p_l$ be distinct prime ideals of $\zf$ coprime to the discriminant of $\qa$ and let $I$ be a primitive right ideal of $\mo$ with $\nrd\at{I} = \p_1^{k(1)} \ldots \p_l^{k(l)}$.  Then $I$ factors uniquely into a proper product 
\begin{equation}\label{eq:ordered}
 I = I_{1,1} \cdots I_{1,k(1)} I_{2,1} \cdot I_{2,k(2)}\cdots I_{l,k(1)}\cdots I_{l,k(l)},
\end{equation}
of maximal ideals such that $\nrd{I_{i,j}} = \p_i$. 
\end{prop}
\begin{proof} \label{prop:unique-factor}
The existence of the factorization follows from \corr{ordering}. The total number of products of the form \eq{ordered} that correspond to a primitive ideal is at most 
\[
 N\at{\p_1}^{k(1)-1}\at{N\at{\p_1}+1}  \cdots   N\at{\p_l}^{k(l)-1}\at{N\at{\p_l}+1} .
\]
On the other hand, by equation \eq{degree}, there are precisely that many distinct primitive two-sided ideals with  norm $\p_1^{k(1)} \cdots \p_l^{k(l)}$. We conclude that no two products of the form \eq{ordered} can be equal and therefore the factorization of any integral ideal of the form \eq{ordered} is unique. 
\end{proof}

The proposition above allows us to define a spanning tree of $\gs$. We first note that path connecting root $\mo$ and any vertex $\mo'$ of $\gs$ corresponds to the factorization of the unique primitive ideal containing $\mo'\mo$ into primitive maximal ideals. We now restrict $\gs$ to a subgraph that contains only path described in the proposition above. This subgraph is connected because described type of factorization exists and there is a unique path between any vertex and the root. We conclude that the subgraph is a tree We call it a spanning tree of $\gs$ and use notation $\tsm$ for it. 

It is useful to observe that the product of two primitive ideals is not always primitive. The following proposition shows under which condition it is true. The following  proposition will be useful for understanding canonical forms of factorizations. It also gives an idea why formula \eq{degree} has separate factors for different exact prime factors. 

\begin{prop} \label{prop:product} Let $I$ and $J$ be primitive ideals such that $\ordr{I}=\ordl{J}$ and $\nrd\at{I}$ is coprime to $\nrd\at{J}$, then $IJ$ is a primitive ideal. 
\end{prop}
\begin{proof} 
Suppose that $IJ$ is not primitive. Then there must exist a prime ideal $\P$ of $\ordr{J}$ such that $IJ \subset \P$. Without loss of generality, assume that $\nrd\at{\P}$ divides $\nrd\at{I}$. Let $J_n \ldots J_1$ be a factorization of $J$ into a product of maximal ideals. Let $\P'$ be the prime ideal of $\ordl{J_1}$ contained in $J_1$ and let $J_1^{c}$ be the integral ideal such that $ J_1 J_1^{c} = \P'$. Note that $IJ J_1^{c} \subset \P'$ and $IJ J_1^{c} \subset \P J_1^{c} = J_1^{c} \at{J_1^{c}}^{-1} \P J_1^{c} \subset \at{J_1^{c}}^{-1} \P J_1^{c}$. We conclude that $IJJ_1^{c} \subset \P' \cap \at{J_1^{c}}^{-1} \P J_1^{c} = \P'\at{J_1^{c}}^{-1} \P J_1^{c}$. The equality $\P' \cap \at{J_1^{c}}^{-1} \P J_1^{c} = \P'\at{J_1^{c}}^{-1} \P J_1^{c}$ is true because $\P'$  and $\at{J_1^{c}}^{-1} \P J_1^{c}$ have different norms and therefore different prime ideals. Above implies that $I J_n \ldots J_2 \subset \at{J_1^{c}}^{-1} \P J_1^{c}$. We proceed in a similar way to show that $I \subset \P''$ for some two-sided ideal $\P''$. This contradicts to the assumption that $I$ is a primitive ideal. 
\end{proof}

Proposition above also allows one to find relation between different elements of $\gens$. Let us consider an example to illustrate the idea. Let $S = \set{\p_1,\p_2}$ where both $\p_1$ and $\p_2$ does not divide the discriminant of the algebra. Suppose also that all one sided ideals are principal. Let $g_1$ and $g_2$ be elements of $\gens$ whose norms generate ideals $\p_1$ and $\p_2$. Consider now primitive ideal $g_1 g_2 \mo$ that factors into maximal ideals $I_1 I_2$ with norm $\p_1$ and $\p_2$. However there exist factorization $I'_2 I'_1$ such that $\nrd\at{I'_2} = \p_2, \nrd\at{I'_1} = \p_1$. Therefore we find that $g_1 g_2 = g'_2 g'_1 u$ for some unit $u$ and $g'_1, g'_2$ where $\nrd{g'_k}$ generate ideal with norm $\p_k$. This idea also can be used to find relations with more then two generators.  

\subsubsection{The algorithm for constructing spanning tree of $\gs$. Estimating its depth and size. }

As we discussed before, \propos{unique-factor} implies that any ordering of the elements of $S$ coprime to the discriminant of $\qa$ defines the spanning tree $\tsm$ of $\gs$. Now we discuss the algorithm for building the spanning tree of $\gs$ and finding elements of $\gens$. We first give a formal definition of the data structure we use to store $\tsm$ and then describe its interpretation.

\begin{dfn} \label{defn:tree-desc} Let the description of $\tsm$ be a sequence $V_j, j = 0, \ldots, M$ where $M$ is possibly $\infty$. Each $V_k$ is a finite sequence $v_{k,1},\ldots,v_{k,M(k)}$. Each $v_{k,j}$ is a tuple $(s_1,s_2,s_3,q)$ such that $s_1,s_2,s_3$ are integers and $q$ is from $\qa$. Let also 
\begin{enumerate}
	\item IDEAL-ID$(s_1,s_2,s_3,q)=s_1$
	\item PARENT-ID$(s_1,s_2,s_3,q)=s_2$
	\item MAX-ORDER$(s_1,s_2,s_3,q)=(s_3,q)$
	\item $V[k]=V_k$, $V[k,j] = v_{k,j}$ 
\end{enumerate}
We call $V$ the sequence of levels of $\tsm$, $V_j$ a level of $\tsm$ and $v_{i,j}$ a vertex description. $V_0$ consists of the single element $(1,0,k_0,q_0)$.
\end{dfn} 

By constructing $\tsm$ we mean first deciding if $M$ is finite and then finding the description of $\tsm$ described above.

To interpret the definition above we fix the ordering of ideals from $S$ coprime to discriminant of $\qa$ and conjugacy classes of maximal orders of $\qa$. Recall that the number of conjugacy classes of maximal orders of $\qa$ is always finite. Let $\p_1,\ldots,\p_l$ be ideals from $S$ coprime to discriminant of $\qa$. Let $\mo_1,\ldots,\mo_m$ be representatives of all conjugacy classes of maximal orders of $\qa$. 

The vertex description $v$ defines a vertex of $\tsm$ which is a maximal order. Let MAX-ORDER$\at{v}=\at{s,q}$, then the corresponding maximal order is $q\mo_s q^{-1}$. This representation is valid because for any maximal order of $\qa$ we can find such pair $\at{s,q}$. In particular, MAX-ORDER$\at{v_{0,1}}=\at{s_0,q_0}$ where $\mo = q_0 \mo_{s_0} q_0^{-1}$. For all $v$ from $V_k$ it is the case that MAX-ORDER$\at{v}$ correspond to maximal orders $\mo'$ such that $\cmi{\mo,\mo'}=k$. 

The vertex description also contains information about edges of $\tsm$. For $j=\text{PARENT-ID}\at{v}$ the maximal order $\mo'$ described by $v_{k-1,j}$ is the unique $S$-neigbour of $q\mo_s q^{-1}$ such that $\cmi{\mo',\mo}=k-1$. More precisely,  $q\mo_s q^{-1}$ is a $\p$-neighbor of $\mo'$, where $\p$ has index IDEAL-ID$\at{v}$ in the sequence $\p_1,\ldots,\p_l$. If the vertex description sequence $v_1,\ldots,v_n$ corresponds to a path in $\tsm$, it must be the case that IDEAL-ID$\at{v_1},\ldots,$ IDEAL-ID$\at{v_n}$ is a 
non-decreasing sequence of integers. This is because every such path corresponds to a factorization of a primitive ideal into a sequence of the form \eq{ordered}. 

The algorithm for building $\tsm$ relies on the description of the $S$-neighbors of all conjugacy classes of maximal orders given by the {\it maximal orders adjacency structure} (see \defin{adj-desc} ). It is a part of the procedure FIND-S-GENERATORS on \fig{gs-alg}. 

We do not provide pseudo-code of some procedures used in pseudo-code as algorithms implementing them are well-known and their implementations are available. For example, many of them are part of software package MAGMA. Below is the list of this procedures together with references describing corresponding algorithms:

\rem{Procedures used in pseudo-code}\label{rem:procedures}

\begin{itemize}

	\item Finding generators of a maximal order (\fig{neigbours}, line \ref{line:gens}). Implementation is available in MAGMA. 
	\item Testing the membership in a maximal order (\fig{neigbours}, line \ref{line:membership}). Implementation is available in MAGMA. 
	\item Computing the discriminant of quaternion algebra (\fig{tree-size}, line \ref{line:disc}). Implementation is available in MAGMA. 
	\item LENGTH. Returns length of the sequence. 
	\item TOTAL-IDEALS. Number of ideals in the set $S$ that does not divide the discriminant. Testing if given ideal divides discriminant of the algebra can be preformed using MAGMA.
	\item For MAX-ORDER, PARENT-ID, IDEAL-ID see \defin{tree-desc}.
	\item We compute valuation $v_S\at{x} = \sum_{\p \in S}v_{\p}\at{x \zf}$, where $v_{\p}$ is a $\p$-adic valuation of the ideal $x\nf$ of number field $F$. Algorithms for computing $v_{\p}$ are well-known.  
	\item  Computing the generators of prime ideals whose norm divides the discriminant of $\qa$ (\fig{exact-synthesis}, line \ref{line:prime-gens}). The algorithm is well known \cite{KV}, its implementation is available in MAGMA.  
	\item IDEAL-GENERATOR (\fig{exact-synthesis}, line \ref{line:ideal-gen}). Finds a generator of a right principal $\mo$-ideal. The algorithm is well known \cite{KV} and is implemented in MAGMA.  
	\item UNIT-DECOMPOSE (\fig{exact-synthesis}, line \ref{line:units}). Decomposes a unit of a maximal order $\mo$ into a product of generators of $\mo^\times$. See \sec{unit-groups} for the related discussion. 
	\item Algorithms for conjugation of maximal orders by an element, multiplication of ideals and maximal orders, ideals inversion, conjugation of ideals by quaternion are well known. (\fig{exact-synthesis}, line \ref{line:multiply})
\end{itemize}

\begin{figure}
\begin{algorithmic}[1]
\Require Maximal order $\mo$ of $\qa$, set $S$ of prime ideals of $\zf$
\Procedure{FIND-S-GENERATORS}{$\mo$,$S$}
\State $M \gets $\Call{SPANNING-TREE-SIZE}{$\mo$,$S$}
\If{$M = \infty $}
\State \Return $\varnothing$
\Else
\State $\Adj \gets $\Call{MAX-ORDERS-ADJ}{$\mo$,$S$}
\State $\at{s_0,q_0} \gets$ \Call{CONJ-CLASS-DESCR}{$\mo$}
\State $V \gets$ sequence of levels of $\tsm$ of size $M$, each level is empty
\State $V_0 \gets \at{1,0,s_0,q_0}$ \label{line:level-zero} \Comment Sequence of length one
\State $V_1 \gets $ \Call{S-NEIGHBORS}{$\Adj$,$1$,$V_0[1]$} \Comment See \fig{neigbours} \label{line:level-one}
\ForAll{ $k = 2,\ldots,M$ } 
\State $V_k \gets $ \Call{NEW-LEVEL}{$\Adj$,$s_0$,$V_{k-1}$,$V_{k-2}$} \label{line:level-k}
\EndFor
\State \Return $\set{ q : \at{s_0,q} = \text{MAX-ORDER}\at{v}, v \in V_k, k= 1,\ldots,M }$
\EndIf
\EndProcedure
\Ensure Set of canonical generators $\gens$ if $\gs$ finite and $\varnothing$ otherwise
\Statex

\Require Maximal orders adjacency description $\Adj$, $s_0$ -- conjugacy class of $\mo$,
\Statex vertices description of the current layer $V_{c}$, of the previous layer $V_{pr}$. 
\Procedure{NEW-LEVEL}{$\Adj$,$s_0$,$V_{c}$,$V_{pr}$}
\State $V_{new} \gets $ empty sequence of vertex descriptions
\ForAll{ $m$ = $1,\ldots,$ \Call{LENGTH}{$V_c$}}
\State $\at{s',q'} \gets \text{MAX-ORDER}\at{V_c[m]}$
\If{ $s' \ne s_0$ } 
\State $N \gets $\Call{S-NEIGHBORS}{$\Adj$,$m$,$V_c[m]$} \Comment See \fig{neigbours}
\State $O \gets$ \Call{MAX-ORDER}{$V_{pr}[\text{PARENT-ID}\at{V_c[m]}]$} 
\ForAll{ $v' \in N $}
\If{ \Call{IDEAL-ID}{$v$} $=$ \Call{IDEAL-ID}{$v'$} }
\If{ $\text{IS-EQUAL} \at{O,\text{MAX-ORDER}\at{v'}}$ } \Comment See \fig{neigbours}
\State Append $v'$ to $V_{new}$	\label{line:add-sometimes}
\EndIf
\Else
\State Append $v'$ to $V_{new}$	\label{line:add-always}
\EndIf
\EndFor
\EndIf
\EndFor
\State \Return $V_{new}$
\EndProcedure

\end{algorithmic}

\caption{\label{fig:gs-alg} The algorithm for constructing spanning tree of $\gs$. }
\end{figure}

\begin{figure}
\begin{algorithmic}[1]
\Require Maximal orders adjacency description $\Adj$, 
\Statex index of the vertex description in its level $m$, 
\Statex vertex description $v$
\Procedure{S-NEIGHBORS}{$\Adj$,$m$,$v$}
\State $\at{s,q} \gets$ \Call{MAX-ORDER}{$v$} 
\State $N \gets$ empty sequence of vertex descriptions 
\ForAll{$k=\text{IDEAL-ID}\at{v},\ldots,\text{TOTAL-IDEALS}$}  \Comment See \rema{procedures}
	\State Append $ v \in \set{\at{k,m,s',qq'} : \at{s',q'} \in \Adj_{s,k} }$ to $N$ \label{line:s-neighbors}
\EndFor
\State \Return $N$
\EndProcedure 
\Ensure Vertex descriptions of $S$-neighbours of $v$

\Statex 
\Require $s_i,q_i$ defining a maximal order $q_i \mo_{s_i} q_i^{-1} $
\Statex $\mo_{s_i}$ is the maximal order with index $s_i$ in the 
\Statex sequence output by \Call{CONJ-CLASSES-LIST}{}  \Comment See  \rema{procedures}
\Procedure{IS-EQUAL}{$\at{s_1,q_1}$,$\at{s_2,q_2}$}
\If{$s_1 = s_2$}
\State $\mo_1, \ldots, \mo_l \gets$ \Call{CONJ-CLASSES-LIST}{}  \Comment See  \rema{procedures}
\State Let $g_1,\ldots,g_m$ be generators of $\mo_{s_1}$ \label{line:gens}
\State $q \gets q^{-1}_1 q_2$
\State \Return TRUE if all $q g_k q^{-1}$ are from $\mo_{s_1}$, FALSE otherwise \label{line:membership}
\Else 
\State \Return FALSE
\EndIf
\EndProcedure 
\Ensure TRUE if maximal orders are equal and FALSE otherwise

\Statex 
\end{algorithmic}
\caption{\label{fig:neigbours} Subroutines used in the algorithm for building spanning tree of $\gs$ and finding its leaves. }
\end{figure}

\begin{figure}
\begin{algorithmic}[1]
\Statex 
\Require Maximal order $\mo$ of $\qa$, set $S$ of prime ideals of $\zf$
\Procedure{SPANNING-TREE-SIZE}{$\mo$,$S$}
\State $\Adj_{1,1},\ldots,\Adj_{m,l} \gets$ \Call{MAX-ORDERS-ADJ}{$S$}
\State $\at{s_0,q_0} \gets$ \Call{CONJ-CLASS-DESCR}{$\mo$}
\State $E_1 \gets \set{ \at{s_0, j, i} : \at{j,q} \in \Adj_{s_0,i}, i=1,\ldots,l, j \ne s_0 }$ \label{line:e1}
\State $k \gets 1$
\While{$E_k$ is not empty}
\ForAll{$\at{s,j,i} \in E_k$ }
\If{$\#\set{\at{s,q} : \at{s,q} \in \Adj_{j,i}, s \ne s_0 } > 1 $}
\State $E_{k+1} \gets E_{k+1} \cup \set{\at{j,s,i} } $ \label{line:group-one-a}
\EndIf
\State $E_{k+1} \gets E_{k+1} \cup \set{ \at{j,s',i} : \at{s',q} \in \Adj_{j,i'}, s' \ne s_0, s' \ne s } $ \label{line:group-one-b}
\ForAll{$i' = i+1,\ldots,l$}
\State $E_{k+1} \gets E_{k+1} \cup \set{\at{j,s',i'} : \at{s',q} \in \Adj_{j,i'}, s' \ne s_0 } $ \label{line:group-two}
\EndFor
\EndFor
\If{ $E_{k+1}$ is equal to one of $E_1,\ldots,E_k$ }
\State \Return $\infty$
\Else
\State $k \gets k + 1$
\EndIf
\EndWhile
\State \Return $k$
\EndProcedure 
\Ensure Depth of the $\tsm$ if its finite and $\infty$ otherwise 
\end{algorithmic}
\caption{\label{fig:tree-size} The algorithm for finding depth of the spanning tree of $\gs$}
\end{figure}

\begin{figure}
\begin{algorithmic}[1]
\Require Maximal order $\mo$ of $\qa$, set $S$ of prime ideals of $\zf$,
\Statex an element $q$ of $\mo$ such that $\nrd\at{q}\zf$ factors into ideals from $S$ 
\Statex (assumes that two side ideal class group of $\qa$ is trivial)
\Procedure{EXACT-SYNTHESIS-1}{$\mo$,$S$,$q$}
\State $\gens \gets$\Call{FIND-S-GENERATORS}{$\mo$,$S$} \Comment Precomputed, \fig{gs-alg}
\State $G \gets \set{ \text{PRIMITIVE-REPR}\at{\mo,x} : x \in \gens }  $ \Comment Precomputed 
\State Let $S_0 = \set{\p_1,\ldots,\p_l}$ be all elements of $S$
\Statex that does not divide  the  discriminant of $\qa$ 
\State $q \gets \text{PRIMITIVE-REPR}\at{\mo,q}$
\State $C \gets $ empty sequence of quaternions 
\While{$v_{S_0}\at{\nrd\at{q}} > 0 $}
\State Find $q_{\min}$ from $G$ such that $q^{-1}_{\min}q$ is in $\mo$
\Statex and $v_{S_0}\at{\nrd\at{q^{-1}_{\min}q}}$ minimal possible 
\State Add $q_{\min}$ to the end of $C$,  $q \gets q^{-1}_{\min}q$
\EndWhile
\State Add \Call{TWO-SIDED-DECOMPOSE}{$q$,$\mo$,$S$} to the end of $C$
\State \Return $C$
\EndProcedure 
\Ensure Outputs $q_1, \ldots, q_m, u_1, \ldots, u_n, \alpha $ such that there product is $q$,
\Statex  where $q_i$ are from $\gens \cup \gtwo$ 
\Statex and $u_j$ are generators of the unit group of $\mo$
\Statex

\Procedure{TWO-SIDED-DECOMPOSE}{$q$,$\mo$,$S$}
\Statex (assumes that two side ideal class group of $\qa$ is trivial)
\State $Q_1,\ldots,Q_M$ - generators of prime ideals with the norm 
\Statex that divides discriminant of $\qa$ and from $S$ \label{line:prime-gens}
\State $C \gets $ empty sequence of quaternions 
\State Find $\alpha$ from $\nf$ such that $\alpha q$ is integral and
\Statex $\nrd\at{\alpha q}$ does not divide the discriminant of $\qa$, $q \gets \alpha q$
\While{$v_{S}\at{\nrd\at{q}} > 0 $}
\State Find $q_{\min}$ from $Q_1,\ldots,Q_M$ such that $q^{-1}_{\min}q$ is in $\mo$
\State Add $q_{\min}$ to the end of $C$,  $q \gets q^{-1}_{\min}q$
\EndWhile
\State Add \Call{UNIT-DECOMPOSE}{$q$} to the end of $C$ \label{line:units} \Comment See \sec{unit-groups}
\State \Return $C$, $\alpha$
\EndProcedure  
\Statex

\Require Maximal order $\mo$ of $\qa$, an element $q$ of $\mo$ 
\Statex (assumes that two side ideal class group of $\qa$ is trivial)
\Procedure{PRIMITIVE-REPR}{$\mo$,$q$}
\State $I \gets \at{q\mo q^{-1}}\mo $ \label{line:multiply}
\State $I_{pr} \gets \text{PRIMITIVE-IDEAL}\at{I}$ \Comment See \fig{prim}
\State \Return $\text{IDEAL-GENERATOR}\at{I}$ \label{line:ideal-gen} \Comment See \rema{procedures}
\EndProcedure
\Statex

\end{algorithmic}
\caption{\label{fig:exact-synthesis} Exact synthesis algorithm for finding decomposition described in \theo{factor-2} and related subroutines. }
\end{figure}

Now we discuss the details of the FIND-S-GENERATORS procedure. Its correctness immediately follows from the discussion. The procedure builds levels (see \defin{tree-desc}) of the spanning tree one by one. Level zero (line \ref{line:level-zero}) contains only one vertex corresponding to $\mo$. We computed the description of $\mo$ using CONJ-CLASS-DESCR. Level one (line \ref{line:level-one}) is constructed using the procedure S-NEIGHBORS (\fig{neigbours}),  which is also a core subroutine for constructing all subsequent levels. We discuss it in more detail next. 

Given vertex description $v$ (see \defin{tree-desc}) and maximal orders adjacency structure (see \defin{adj-desc}), the procedure S-NEIGHBORS constructs a list of vertex descriptions of S-neighbors of the maximal order $\mo'$ corresponding to $\at{s,q}=$MAX-ORDER$\at{v}$. More precisely, it gives a list of all $\p_k$ neighbors of $\mo'$ with $k \ge $ IDEAL-ID$\at{v}$. We get the description of all $\p_k$ neighbors of $\mo'$ using the maximal order adjacency structure. Given a list 
$
 \Adj_{s,k} = (s_1,q_1),\ldots,(s_{N},q_N)
$
of descriptions of $\p_k$-neighbors of $\mo_s$, the list of its $\p_k$-neighbors is
$
q_1 \mo_{s_1} q^{-1}_1, \ldots, q_N \mo_{s_N} q^{-1}_N
$. 
As $\mo' = q\mo_s q^{-1}$ we see that all $\p_k$-neighbors of $\mo'$ are of the form  $q q_j \mo_{s_j} q^{-1}_j q^{-1}$ and their description is $\at{s_j,q q_j}$. This is precisely how we construct the sequence $N$ on line \ref{line:s-neighbors} of the S-NEIGHBORS procedure. Level one of the spanning tree must include all $S$-neighbors of $\mo$. This is precisely what happens because we have IDEAL-ID$\at{V_0[1]}$ equal to one. Currently we have shown that our algorithms builds levels zero and one of the spanning tree correctly. 

Let us now discuss the procedure NEW-LEVEL that builds level $k$ of the tree given levels $k-1$ and $k-2$. We will show that it correctly builds level $k$, under the assumption that levels $k-1,k-2$ were built correctly. To build a new level we iterate through all vertices $v$ of level $k-1$ that are not leaves (in other words, that are not of conjugacy type $s_0$). For each $v$ and for each $j\ge$ IDEAL-ID$\at{v}$, we find the $\p_j$-neighbors of the maximal order represented by $v$. They are candidates for elements of level $k$. The extra constraint on $j$ ensures that we will only get paths in the tree that correspond to factorizations with ordered norms of factors. By \propos{connect} we know that there are two alternatives for $\p_j$-neighbors of vertices in level $k-1$: They are either in level $k-2$ or in level $k$. When $j>$ IDEAL-ID$\at{v}$, the only possibility is that the $\p_j$-neighbors are in level $k$. 
This follows from \propos{product}. This is why in this case we add the corresponding vertex description to $V_{new}$ on line \ref{line:add-always} without extra checks. When $j$=IDEAL-ID$\at{v}$, we need to verify that the order we are adding belongs to the level $k$. Because of the uniqueness of the factorization, the only order that is (a) a $\p_j$-neighbor of the order described by $v$ and (b) belongs to level $k-2$ is the order $\mo''$ corresponding to PARENT-ID$\at{v}$. 
In the procedure NEW-LEVEL, we add to $V_{new}$ precisely the orders not equal to $\mo''$ on line  \ref{line:add-sometimes}. We have shown that NEW-LEVEL only adds orders to level $k$ that belong to level $k$. It remains to show that it finds all of them. Suppose the order $\mo'$ should be on level $k$; then there is a unique ordered factorization of the primitive ideal that contains $\mo'\mo$. This immediately implies that $\mo'$ is a $\p_j$-neighbor of some order on level $k-1$ and that the IDEAL-ID of the corresponding vertex is less than $j$. This completes the proof of the correctness of NEW-LEVEL. 

Now we discuss the SPANNING-TREE-SIZE procedure that determines the number of levels in the spanning tree. In the SPANNING-TREE-SIZE procedure, we compute a rough description of all edges of $\tsm$. Now we describe what this means precisely. Note that by \propos{connect} every edge of $\tsm$ connects maximal orders of level $k-1$ and $k$. Let $\mo'$ be the maximal order of level $k-1$ and $\mo''$ be the maximal order of the level $k$. With each edge, we associate a tuple $\at{s,j,i}$ that we call an edge description. Integers $s$ and $j$ denote the conjugacy class of $\mo'$ and $\mo''$. Integer $i$ in the description means that $\mo'$ and $\mo''$ are $\p_i$-neighbors. The sets $E_k$ that we build in the algorithm are sets of all edge descriptions corresponding to the edges connecting levels $k-1$ and $k$, excluding edges that are connected to leaves on level $k$. The core observation is that we can build $E_{k+1}$ from $E_{k}$, only using the maximal orders adjacency structure $\Adj$. The spanning tree $\tsm$ has depth $N$ if and only if $E_{N}$ is empty but $E_{N-1}$ is not. Also note that the number of different edge descriptions is finite (and completely described by $\Adj$), so therefore the number of possible sets $E_k$ is also finite. This implies that $\tsm$ has infinite depth if and only if there exist $k$ and $n$ such that $E_n = E_k$.

Now we discuss how to build $E_1$ and how to build $E_{k+1}$ from $E_{k}$. Elements of $E_1$ have the form $\at{s_0,j,i}$, where $j$ goes through the labels of all conjugacy classes not equal to $s_0$ of $\p_i$-neighbors of $\mo$, $s_0$ is a conjugacy class of $\mo$ and there is no restriction on $i$. This is precisely how we build $E_1$ in SPANNING-TREE-SIZE on line \ref{line:e1}. Let us now look how to build $E_{k+1}$ given $E_{k}$. We iterate through all edge descriptions $E_k$. Each edge description $\at{s,j,i}$ corresponds to all possible pairs of orders $\at{\mo',\mo''}$ such that: 
\begin{itemize}
	\item $\mo'$ belongs to level $k-1$ of $\tsm$, conjugate to $\mo_s$ 
	\item $\mo''$ belongs to level $k$ of $\tsm$, conjugate to $\mo_j$ 
	\item $\mo'$ and $\mo''$ are $\p_i$-neighbors 
\end{itemize}
There are two sets of maximal orders that are in level $k+1$ and are $S$-neighbors of $\mo''$. One set consists of the $\p_i$-neighbors of $\mo'$ and the other set contains the $\p_{i'}$-neighbors of $\mo''$ for $i' > i$. Note that number of $\p$-neighbors of $\mo''$ and their conjugacy classes depend only on $j$ and the same for all $\mo''$. For the second group we add edge descriptions $\at{j,s',i'}$ where $s'\ne s_0$ goes through all conjugacy classes of all $\p_{i'}$ neighbors of $\mo''$ (this is done on line \ref{line:group-two}). For the first set we need to take into account that $\mo''$ has one $\p_i$-neighbor in level $k-1$ that is precisely $\mo'$. Therefore we need to add the edge description $\at{j,s,i}$ only in the case when $\mo''$ has more than one $\p_j$-neighbor that is conjugate to $\mo_s$. This this precisely what is done on line~\ref{line:group-one-a}. We also need to add the edge descriptions corresponding to all $\p_i$-neighbors of $\mo''$ that are not conjugates of $\mo_s$. This is what is done on line~\ref{line:group-one-b}. 

The proof of correctness of SPANNING-TREE-SIZE uses the same ideas as the proof of correctness of NEW-LEVEL and we do not provide it here. It is also interesting to note that one can count the number of leaves of $\tsm$ without building it using a slightly modified version of SPANNING-TREE-SIZE that has edge multiplicities in its edge description structure. 

\subsubsection{Canonical form of the factorization}

Now we will briefly discuss why the factorization of elements of $\mos$ considered in \theo{mos-factor} corresponds to a canonical form. Consider $q$ from $\mos$. According to \theo{mos-factor}, we can write $q=q_1 \ldots q_n q_0$, where $q_1,\ldots,q_n \in \gens$ and where $q_0$ generates a two-sided ideal of $\mo$. Let us now assign a label to each $q_k$ from $\gens$. Let $I_k$ be the primitive ideal containing $q_k\mo$. Let $I_{k,1} \cdots I_{k,n(k)}$ be a factorization of $I_k$ into primitive maximal ideals. The label of $q_k$ will consist of left and right parts. The left will be the label of the maximal ideal $I_{k,1}$, and the right part will be the label of the maximal ideal $I_{k,n(k)}$. The label of a maximal ideal consists of two parts. The first part is the index of its norm $\p$ in the list $\p_1,\ldots,\p_l$ of elements of $S$ that do not divide the discriminant of $\qa$. The second is the point of the projective line $\pone\at{\zf/\p}$ corresponding to the ideal by \lemm{prim-enum}. Let us order the factorization of the primitive ideal that contains $q\mo$ according to \propos{ordered} and build the factorization $q_1 \ldots q_n q_0 $ out of it. Let $\at{\at{i_k,(x_k,y_k)},\at{j_k,(z_k,w_k)}}$ be a label of $q_k$. It is not difficult to see that $j_k \le i_{k+1}$. In the case when $j_k = i_{k+1}$, it must be the case that $(x_k,y_k) \ne (z_k,w_k)$. On the other hand, if we obtained two different factorizations in terms of $\gens$ satisfying the above-mentioned condition on labels, then they must correspond to two different elements of $\mos$. This follows from the uniqueness of the ordered factorization. Finally, for any factorization in terms of $\gens$ that satisfies the label constraints, it is true that $\cmi{q_1\ldots q_n q_0 } = \sum_{k=1}^n \cmi{q_k}$. Later we will see that these generic ideas explain some of the canonical forms of quantum circuits that can be found in the literature. 


\section{Applications} \label{sec:app}


The goal of this section is to explain how to use tools described in the previous section to study the following questions about two by two unitaries: 
\begin{itemize}
	\item Given some finite set $G$ of unitaries, find a special form of unitaries containing the group they generate. 
	\item Are all unitaries of this special form representable by elements from the set $G$? If so, how can we find  such a representation. If not, 
	\begin{itemize}
		\item How to find the counterexample?
		\item Can one find a bigger set $G'$ of unitaries that generates all matrices of the mentioned special form?  
	\end{itemize} 
	\item What is a good general approach to specifying this special form of the unitaries? 
	\item Is there some canonical set $G''$ of unitaries  such that any unitary of the mentioned special form can be uniquely written using elements of $G''$? Is there a canonical form for unitaries with this special form?
\end{itemize}
These questions are usually studied under the name of exact synthesis of unitaries. Our goal is to provide guidance for choosing a number field $F$, a quaternion algebra $\qa$, a maximal order $\mo$ in the quaternion algebra and the set $S$ of ideals of $\zf$ to study some interesting instances of exact synthesis problems. Roughly speaking, all these choices depend on the number field over which the unitary matrices are represented. It is also much simpler to map special unitaries to quaternions than to map general unitaries to quaternions. However, the simple way of turning a general unitary $U$ into a special unitary $U / \sqrt{\det U}$ is not applicable in our case, as taking square roots generally makes number fields more complicated. In this section, we first discuss an alternative, unnormalized realization of unitaries by quaternions that enables the use of simpler number fields.  Then, we discuss how to make the choices mentioned above by working with the unnormalized realization of unitaries from the gate set we study or related to the special form of the unitaries we study.

\subsection{Choosing a quaternion algebra using the unnormalized realization of unitaries } Let $K$ be a CM field (i.e.\ a totally imaginary extension of a totally real number field) with real subfield $\nf = K \cap \r$. Consider the following form of unitary matrices: 
\[
U = \frac{1}{\sqrt{d}}\UG{x}{z\sqrt{b}}{y\sqrt{b}}{w}, x,y,z,w \in K, b,d \in \nf.  
\]
Most unitaries considered in the applications have the form described above. We can rewrite $U$ as a product  $U_1 U_2 $ of two unitaries, where 
\[
U_1 = \frac{1}{\sqrt{d}}\UG{x}{-y^{\ast}\sqrt{b}}{y\sqrt{b}}{x^{\ast}}, U_2 = \UG{1}{0}{0}{\alpha}, \alpha = \frac{\det U}{\det U_1} \in K.
\]
Now we will show that we can represent $U$ using the set of matrices 
\[
M\of{K,b} = \set{ \UG{x}{-y^{\ast}\sqrt{b}}{y\sqrt{b}}{x^{\ast}} : x,y \in K, b \in F }.
\]
We chose $b$ to be a part of the definitions of the set of unitaries because in most applications $\sqrt{b}$ is usually fixed for the gate set that one studies. We say that the matrix $M$ from $M\of{K,b}$ represents $U$ if the following holds for all density matrices $\rho$: 
\[
	U\rho U^{\dagger} = \frac{M \rho M^\dagger}{ \det M}.
\]
Note that if $M_1$ represents $U_1$ and $M_2$ represents $U_2$ then $M_1 M_2$ represents $U_1 U_2$. It is not difficult to see that when $\al \neq -1$, the above unitaries $U_1$ and  $U_2$ are represented by
\[ 
M_1 =  \UG{x}{-y^{\ast}\sqrt{b}}{y\sqrt{b}}{x^{\ast}}, M_2 = \UG{1+\alpha^{\ast}}{0}{0}{1 + \alpha}. 
\]
Note that by using this realization, we got rid of $\sqrt{d}$ and the matrices we consider are special unitaries up to a scalar. Another useful property of the unnormalized representation is that $M$ and $\beta M$ for $\beta \in F$ represent the same unitary matrix. 

Now we show how to choose a quaternion algebra corresponding to $M\of{K,b}$. Note that $K$ can always be represented as $F(\sqrt{-D})$, where $D$ is a totally positive element of $F$. Any element of $K$ can be written as $x_1+x_2\sqrt{-D}$ for $x_1,x_2$ from $F$. Using $\sqrt{-D}$, an unnormalized realization of the Pauli matrix $Z$ (the case $\alpha=-1$ excluded above) can be written as
\[
\UG{\sqrt{-D}}{0}{0}{-\sqrt{-D}}.
\]
An arbitrary matrix $M$ from $M\of{K,b}$ can be written as 
\[ 
 \UG{x_1 + \sqrt{-D}x_2}{-\at{y_1 - \sqrt{-D}y_2}\sqrt{b}}{\at{y_1 + \sqrt{-D}y_2} \sqrt{b}}{x_1 - \sqrt{-D}x_2}, \,\,\,\,x_1,x_2,y_1,y_2 \in F.
\] 
Using Pauli matrices $I,X,Y,Z$, we can write a matrix from $M\of{K,b}$ as
\begin{equation}
 x_1 I + x_2 \sqrt{-D}Z - y_1 \sqrt{-b}Y + y_2 \sqrt{-Db} X. \label{eq:pauli-repr}
\end{equation}

 We now map elements of $M\of{K,b}$ to the quaternion algebra $Q = \at{\frac{-D,-b}{F}}$. Let $\i,\j,\k$ be elements of $\qa$ such that $\i^2 = -D, \j^2 = -b$ and $\k = \i\j$. Define an $F$-linear map $\kappa \colon M\of{F,b} \to Q$ via
 \[ 
 \kappa\at{\sqrt{-D}Z} = \i, \,\,\,\, \kappa\at{-\sqrt{-b}Y} =\j,\,\,\,\, \kappa\at{\sqrt{-Db} X} = \k.
 \]
 The quaternion $\kappa\at{M}$ corresponding to the matrix $M$ in \eq{pauli-repr} is
$
 x_1 + x_2 \i + y_1 \j + y_2 \k
$. 
The map $\kappa$ has several useful properties. The determinant of $M$ is equal to the norm of $\kappa\at{M}$. The trace of $M$ is equal to the reduced trace of $\kappa\at{M}$. The conjugate $x_1 - x_2 \i - y_1 \j - y_2 \k$ of the quaternion $\kappa\at{M}$ is equal to $\kappa\at{M^{\dagger}}$. 

We have shown how to establish the correspondence between unitaries and quaternions using an unnormalized realization. It is useful to note that the inverse of a unitary $U$ is represented by a conjugate of the quaternion corresponding to $U$. Table~\ref{tab:qa} shows the quaternion algebras and number fields corresponding to some well known gate sets. Some of our examples are related to cyclotomic number fields $\q\at{\zeta_n}$ where $\zeta_n$ is $n$-th root of unity. The real subfield in this case is $\q\at{\zeta_n+\zeta^{-1}_n}$ when $n\ne 4$ and $\q$ otherwise. Let $D_n$ be such that $\q\at{\zeta_n}=F\of{\sqrt{-D_n}}$ for $F $ being the real subfield. When $4 | n$ we can choose $D_n = 1$, in all other cases we choose $D_n = 4 -\at{\zeta_n + \zeta_n^{-1} }^2$. 

\bgroup
\def\arraystretch{1}
\begin{table}
{\tiny
\begin{centering}
\begin{tabular}{|c|c|c|c|c|c|c|}
\hline 
Gate set & Matrices & Projective      & $K $ & $b$ & $F$ & $D$ \tabularnewline
         & example  &  representation &      &     &     &    \tabularnewline
\hline 
\hline 
Clifford+T & $T=\UG{1}{0}{0}{\zeta_8}$ & $\UG{1+\zeta_8^{\ast}}{0}{0}{1+\zeta_8}$  & $\q\at{\zeta_8}$ & $1$ & $\q\at{\zeta_8+\zeta^{-1}_8}$ & 1 \tabularnewline
           & $H=\frac{1}{\sqrt{2}}\UG{1}{1}{1}{-1}$ & $\UG{i}{i}{i}{-i}$ &  & & & \tabularnewline
\hline 
Clifford- &  $T_n=\UG{1}{0}{0}{\zeta_n}$  & $\UG{1+\zeta_n^{\ast}}{0}{0}{1+\zeta_n}$ & $\q\at{\zeta_n}$ & $1$ & $\q\at{\zeta_n+\zeta^{-1}_n}$ & $D_n$ \tabularnewline
cyclotomic\cite{CYCL}& & & & & & \tabularnewline
\hline 
V-basis\cite{BGS,S1,S2} & $\frac{1}{\sqrt{5}}(I + 2iZ)$ & $ I + 2iZ$   & $\q\at{i}$ & 1 & $\q$ & 1 \tabularnewline
        & $\frac{1}{\sqrt{5}}(I + 2iX)$ & $ I + 2iX$   &  & & &   \tabularnewline
\hline 
Fibonacci\cite{KBS}& $\UG{1}{0}{0}{\zeta_{10}}$  & $\UG{1 + \zeta^{\ast}_{10} }{0}{0}{1 + \zeta_{10}}$  &  $\q\at{\zeta_{10}}$ & $\frac{\sqrt{5}-1}{2}$ & $\q\at{\zeta_{10}+\zeta^{-1}_{10}}$  & $2-b$  \tabularnewline
 & $\UG{b}{-\sqrt{b}}{\sqrt{b}}{b}$ & $\UG{b}{-\sqrt{b}}{\sqrt{b}}{b}$ & & & & \tabularnewline
 & $\UG{1}{0}{0}{-1}$ & $\UG{\sqrt{b-2}}{0}{0}{-\sqrt{b-2}}$ & & & & \tabularnewline
\hline 
$\mathfrak{su}\left(2\right)_{k}$\cite{HZBS}&  $q^{1/4}\UG{-q^{\ast}}{0}{0}{1}$ & $ \UG{1-q^{\ast}}{0}{0}{1-q}$ & $\q\at{q}$ & $\floor{3}_q$ & $\q\at{q+q^{-1}}$ & $D_{k+2}$ \tabularnewline
 &  $\frac{q^{-1/4}}{\floor{2}_q} \UG{q}{-\sqrt{b}}{\sqrt{b}}{q^{\ast}}$ & $\UG{q}{-\sqrt{b}}{\sqrt{b}}{q^{\ast}}$ & & &  & \tabularnewline
 & $\UG{1}{0}{0}{-1}$ & $ \UG{\sqrt{-D}}{0}{0}{-\sqrt{-D}}$ & & & & \tabularnewline
& $q = \zeta_{k+2}$,  & & & & & \tabularnewline
&  $\floor{m}_q = \frac{q^{m/2}-q^{-m/2}}{q^{1/2}-q^{-1/2}} $ & & & & & \tabularnewline
\hline 
$\set{B,K,Z}$ \cite{BKZ} & $K=\UG{1}{0}{0}{\frac{-1+4\sqrt{-3}}{7}}$ & $3I - 2\sqrt{-3}Z$ & $\q\at{\zeta_6}$ & $-1$ & $\q$ & $-3$  \tabularnewline
 \hline 
\end{tabular}
\par\end{centering}
}
\caption{ \label{tab:qa} Examples of choosing quaternion algebra for different gate sets }
\end{table}

\subsection{ Choosing the maximal order and set of ideals $S$ } 
To apply the framework we developed it is crucial to represent unitaries using integral quaternions. A simple way to ensure that the quaternion $x_1 + x_2 \i + x_3 \j + x_4 \k$ is integral is to rescale it by an element $\beta \in \nf$ such that $\beta x_k  \in \zf$. This will ensure that all quaternions corresponding to the gate set of interest belong to the order
\[
L_{\qa} = \zf+\zf \i + \zf \j + \zf \k 
\]
Next we compute a maximal order that contains $L_{\qa}$ using known algorithms. For some classes of quaternion algebras there are known explicit constructions of maximal orders containing $L_{\qa}$ (see e.g.\ Proposition 8.7.3 \cite{Martinet}). Next we discuss a simple generalization of one of examples from  \cite{Martinet}. 

We describe an explicit example of a maximal order in the quaternion algebra $\at{\frac{-1,-1}{F}}$. Suppose $\zf$ contains an element $\xi$ such that $\xi^2$ is equal to $2$ up to a unit of $\zf$. Consider the following order: 
\begin{equation}
 \mo_{F,\xi} = \zf + \frac{\xi}{2} \zf \at{\i+1} + \frac{\xi}{2}\zf \at{\j+1} + \frac{1}{2}\zf \at{1+\i+\j+\k}  \label{eq:max-order}
\end{equation}
By direct computation we can check that the discriminant of the order above is equal to $\zf$. This implies that the order is maximal and the corresponding quaternion algebra has discriminant $\zf$. One example of $F$ that satisfy the required property is $\q\at{\zeta_{8n}+\zeta^{\ast}_{8n}}$ which is directly related to some of the Clifford-cyclotomic gate sets. 

Let us now discuss how to pick the set $S$ of prime ideals of $\zf$ . We assume that we started with some set of gates and now have a set $G_0$ of quaternions in a quaternion algebra $Q$. The set $G_0$ is also a subset of a maximal order $\mo$ we constructed earlier in this section. Given $q \in G_0$, then $q\mo$ is a right ideal of $\mo$. We can write it as 
\begin{equation}
q\mo = I \P_1 \cdots \P_k \cdot \at{\a \mo}, \label{eq:s-choice}
\end{equation}
where $I$ is a primitive ideal, the $\P_j$ are prime ideals of $\CM$ whose norms divide the discriminant of $Q$ and $\a$ is an ideal of $\zf$. If two quaternions have the same factorization \eq{s-choice} up to an ideal $\a$ they correspond to the same unitary and therefore $\a$ is irrelevant for our purposes. We then take $S$ to be the set of prime ideals of $\zf$ that divide the norm of $I$ or that divide the norm of some  $\P_j$. In Table~\ref{tab:s} we provide examples of $S$ for several gate sets of interest. 

\bgroup
\def\arraystretch{1.5}
\begin{table}
\begin{centering}
\begin{tabular}{|c|c|c|c|}

\hline
Gate Set & $\nf$ & $\mo$ & $S$ \tabularnewline
\hline
\hline 
Clifford+T & $\q\at{\sqrt{2}}$ & $\bZ[\sqrt{2}] + \bZ[\sqrt{2}]\frac{1 + \i}{\sqrt 2} +  \bZ[\sqrt{2}]\frac{1 + \j}{\sqrt 2} + \bZ[\sqrt 2] \frac{1 + \i + \j + \k}{2}$ & $ \set{ \mfp_2}$ \tabularnewline
\hline 
V-basis & $\q$ & $\bZ + \bZ\i + \bZ\j + \bZ \frac{1+\i+\j+\k}{2}$ &  $ \set{ 5\z }$ \tabularnewline
\hline 
Clifford+T+V & $\q\at{\sqrt{2}}$ & $\bZ[\sqrt{2}] + \bZ[\sqrt{2}]\frac{1 + \i}{\sqrt 2} +  \bZ[\sqrt{2}]\frac{1 + \j}{\sqrt 2} + \bZ[\sqrt 2] \frac{1 + \i + \j + \k}{2}$ &  $\set{ \mfp_2 , 5\zf }$ \tabularnewline
\hline 
Fibonacci & $\q\at{\varphi},$ & $\bZ[\ph] +  \bZ[\ph] \frac{1+\i}{2} + \bZ[\ph] \j + \bZ[\ph] \frac{\j + \k}{2}$ & $\{\mfp_5\}$ \\
\hline 
$\set{B,K,Z}$ \cite{BKZ} & $\q$ & $\bZ + \bZ\frac{\i+1}{2} + \bZ\j + \bZ \frac{\j + \k}{2}$ & $\set{3\z,7\z}$ \tabularnewline
\hline 
\end{tabular}
\end{centering}
\caption[examples]{ \label{tab:s} Examples of maximal ideals and sets of prime ideals for different gate sets.  Here, $\mfp_2 = (2 + \sqrt 2)\bZ[\sqrt 2] = |1 + \zeta_8|^2 \bZ[\sqrt 2]$ satisfies $\mfp_2^2 = 2 \bZ[\sqrt 2]$, $\ph = \frac{1 + \sqrt 5}{2}$ generates the ring $\bZ[\ph]$ of integers of $\bQ(\sqrt{5})$, and where $\mfp_5 = \frac{5 - \sqrt{5}}{2} \bZ[\ph]$ satisfies $\mfp_5^2 = 5 \bZ[\ph]$.}
\label{table:2}
\end{table}
\egroup 

\subsection{Remarks on previously-studied gate sets} The questions related to the exact synthesis and canonical forms of Clifford+T gate set were widely studied before \cite{GS2,KMM1,GKMR,MA,BS}. Here we briefly summarize the result of applying the framework developed in this paper (see also Tables~\ref{tab:qa},\ref{tab:s}). The quaternion algebra corresponding to Clifford+T is totally definite. The unit group of the maximal order \eq{max-order} is finite and equal to the binary octahedral group, which is otherwise known as the single qubit Clifford group. Interestingly, there is an Euclidean algorithm for this quaternion algebra~\cite{EQ}. This implies that in the Clifford+T case any primitive ideal is principal. Using the generic approach to constructing canonical forms described in the end of \sec{graph}, we precisely obtain the canonical form described in Ref.~\cite{GKMR}. Our framework also leads to a very similar exact synthesis algorithm. Finally, it leads to the same description of exactly synthesizable unitaries. The Clifford+T gate set can be extended by adding single qubit unitaries that can be implemented using RUS circuits \cite{Paetznick2013,RUS1,RUS2}. Our framework can also be applied to such extended gate sets by adding extra prime ideals to the set $S$.

The approach taken in~\cite{KBS} to the exact synthesis of braids for Fibonacci anyons is quite different from the Clifford+T and Clifford-cyclotomic gate sets. This difference is explained by our framework: the relevant quaternion algebra is indefinite and therefore has an infinite unit group. Luckily, in this case the quaternion algebra is defined over the totally real field $\bQ(\sqrt 5)$ and splits at exactly one real place. Therefore, the methods developed in Ref.~\cite{V} apply directly. In this case the set $S$ contains the unique ideal whose norm divides the discriminant of the quaternion algebra. This means that all exact synthesis questions in this case are related to two-sided ideals and the unit group of the maximal order.  

Finally, our methods provide an alternative proof of Theorem~5.1 in~\cite{CYCL}, which characterizes the gates that can be exactly synthesized from some Clifford-cyclotomic gate sets. They also provide an algorithmic solution to the question: What unitaries must one add to the Clifford-cyclotomic gate set to be able to synthesize all unitaries over the ring $\z\of{\frac{1}{2},\zeta_n}$? In the general case, computing the gate set based on the quaternion algebra $\at{\frac{a,b}{F}}$, the maximal order $\mo$ and a set $S$ of prime ideals of $\zf$ is at least as hard as computing the unit group $\zf^\times$~\cite{KV}. The problem of computing the unit group of an arbitrary degree number field is known to be hard for classical computers and solvable on quantum computers~\cite{QUnit}. For small $n$ the question can be answered using computational number theory packages, like MAGMA. For large $n$, the question quickly becomes computationally intractable.

\section{Summary and open questions}

The developed framework allows us to state the question of approximating operators $e^{i\phi Z/2}$ by exactly synthesizable unitaries in a simple way. Consider unitaries corresponding to a quaternion algebra $\at{\frac{a,b}{F}}$, a set of prime ideals $S$ of $\zf$ and a maximal order $\mo$ containing $\zf + \zf \i + \zf \j +\zf \k$. The unitaries can be written  in the form 
\[ 
U = \frac{1}{\sqrt{x_1^2 - a x_2^2 - b x_3^2 + ab x_4^2 }} \at{ x_1 I + \sqrt{a}Z x_2 - \sqrt{b}Y x_3 + \sqrt{ab}X x_4 }, x_j \in \zf,
\]
where the ideal $\a = \at{x_1^2 - a x_2^2 - b x_3^2 + ab x_4^2}\zf$ must factor into ideals from $S$. For example, in the case when $F$ is a totally real number field and $\at{\frac{a,b}{F}}$ is a totally definite quaternion algebra, the integer $\sum_{\p \in S} v_{\p}\at{\a}$ can be used to bound the number of gates required to implement~$U$. This is precisely what is needed for the number-theoretic approximation algorithm: a simple description of the exactly synthesizable unitaries that require a bounded number of gates to be implemented. Developing the details of such an approximation algorithm for definite and indefinite quaternion algebras will be addressed elsewhere~\cite{InPrep}. As soon as an exact unitary is output by an approximation algorithm, one of the exact synthesis algorithms described in this paper can be used to obtain the circuit implementing it. 

There is an interesting open question related to the situation when we do not have all the generators described by quaternion algebra $\at{\frac{a,b}{F}}$, the maximal order $\mo$ and the set $S$. We will illustrate it on a simple example of the $V$-basis (see Table~\ref{table:2}). Consider the gate set $G^{\ast}$ consisting only of $V_0 = \frac{1}{\sqrt{5}}\at{I + 2iX}$ and $ V_1 = \frac{1}{\sqrt{5}}\at{I + 2iZ}$ together with their inverses.  For a binary string $b(1),\ldots,b(n)$, the unitary $U_b = V_{b(1)}\ldots V_{b(n)}$ can be easily decoded using an exact synthesis algorithm. The gate set $G^{\ast}$ is universal and moreover, it is {\it efficiently} universal  (i.e.\ can approximate an arbitrary unitary within precision $\varepsilon$ using a circuit of length $O\at{\log\at{1/\varepsilon}}$) by results contained in \cite{H1,B1}. The same is true for the complete $V$-basis. However, length-$n$ circuits over the full $V$-basis produce a much denser covering of unitaries than is possible with length-$n$ circuits over $G^{\ast}$. Let us now choose $\varepsilon_0$ such that the closed ball $B[U_b,\varepsilon_0]$ of radius $\varepsilon_0$ and with center in $U_b$ contains only one element of $\tpl{G^{\ast}}$ but $2^{\Omega(n)}$ elements representable by circuit of length $n$ over $V$-basis. Let us now pick a random unitary $U'$ inside $B[U_b,\varepsilon_0]$. Is there a polynomial-time algorithm for recovering $b$ from $U'$? It is not difficult to generalize the question to more general gate sets that can be obtained by our framework.  
\bibliography{library}

\end{document}